\documentclass[runningheads,a4page]{llncs}

\usepackage{amssymb}
\usepackage[cmex10]{amsmath}
\usepackage{graphicx}
\usepackage{amsthm}
\usepackage{ifpdf}
\usepackage{color}
\usepackage{algorithm,algorithmic}
\usepackage{fancyvrb}
\usepackage{multirow}
\usepackage{tikz}
\usepackage{stmaryrd}
\usetikzlibrary{arrows,automata}
\usepackage{array}

\newcommand{\restrict}{\upharpoonright}

\renewcommand{\leq}{\leqslant}
\renewcommand{\geq}{\geqslant}
\renewcommand{\emptyset}{\varnothing}
\renewcommand{\le}{\leq}
\renewcommand{\ge}{\geq}

\newcommand{\len}[1]{|#1|}
\newcommand{\states}{\mathcal{S}}
\newcommand{\tree}{T}
\newcommand{\treeNodes}{W}
\newcommand{\atree}[1][]{(\treeNodes_{#1},L_{#1},\func_{#1})}
\newcommand{\alltrees}{\mathbb{T}}
\newcommand{\concat}{\cdot}

\newcommand{\prefixTree}{\preceq}

\newcommand{\closure}{\mathit{cls}}
\newcommand{\powerset}{\mathcal{P}}

\newcommand{\dtmc}{\text{\sf D}}
\newcommand{\AP}{\mathit{AP}}
\newcommand{\dist}{\mathit{Dist}}

\newcommand{\func}{\mathit{P}}
\newcommand{\lastState}[1]{#1\!\!\downarrow}
\newcommand{\support}{\mathit{supp}}
\renewcommand{\P}[2]{\llbracket#2\rrbracket_{#1}}
\renewcommand{\P}[2]{[ #2 ]_{#1}}

\newcommand{\U}{\text{\sf U}}
\newcommand{\X}{\text{\sf X}}
\newcommand{\W}{\text{\sf W}}
\newcommand{\dirac}[1]{\delta_{#1}}

\newcommand{\finitePrefix}{\mathit{Pre}_{\mathit{fin}}}

\newcommand{\Prob}{\Pr}
\newcommand{\comp}[1]{\overline{#1}}
\newcommand{\pctls}{\text{PCTL}_{\mathit{safe}}}
\newcommand{\pctll}{\text{PCTL}_{\mathit{live}}}
\newcommand{\pctlss}{\text{PCTL}_{\mathit{ssafe}}}
\newcommand{\pctlal}{\text{PCTL}_{\mathit{alive}}}
\newcommand{\pctlf}{\text{PCTL}_{\mathit{flat}}}
\newcommand{\weightFunc}[1][R]{\sqsubseteq_{#1}}
\newcommand{\epctls}{\precsim_{\mathit{safe}}}
\newcommand{\epctll}{\precsim_{\mathit{live}}}
\newcommand{\epctllone}{\precsim^1_{\mathit{live}}}
\newcommand{\epctlltwo}{\precsim^2_{\mathit{live}}}

\renewcommand{\bot}{\text{\sf 0}}
\renewcommand{\top}{\text{\sf 1}}
\newcommand{\tco}{\mathit{tco}}
\newcommand{\stsh}{\mathit{ss}}
\newcommand{\natureNum}{\mathbb{N}}

\makeatletter
\newtheorem*{rep@theorem}{\rep@title}
\newcommand{\newreptheorem}[2]{%
\newenvironment{rep#1}[1]{%
 \def\rep@title{#2 \ref{##1}}%
 \begin{rep@theorem}}%
 {\end{rep@theorem}}}
\makeatother

\newreptheorem{theorem}{Theorem}
\newreptheorem{lemma}{Lemma}

\begin{document}
\title{Probably Safe or Live\thanks{This work is supported by the 7th EU Framework Programme under grant
agreements 295261 (MEALS) and 318490 (SENSATION), and by the DFG
Sonderforschungsbereich AVACS. 
Lijun Zhang has received support from the National
Natural Science Foundation of China (NSFC) under grant No.
61361136002 and 91118007.
Joost-Pieter Katoen is supported by the Excellence Initiative of the German 
federal and state governments.}}
\author{
Joost-Pieter Katoen\inst{1}
\and
Lei Song\inst{3,2}
\and
Lijun Zhang\inst{4,2}
}

\institute
{
 \inst{}
  Department of Computer Science, RWTH Aachen University 
  \and
  \inst{}
  Department of Computer Science, Saarland University
 \and
  \inst{}
  Max-Planck-Institut f\"{u}r Informatik
   \and
  \inst{} 
  State Key Laboratory of Computer Science,\\
  Institute of Software, Chinese Academy of Sciences   
}

\authorrunning{Probably Safe or Live}
\maketitle
\begin{abstract}
This paper presents a formal characterisation of safety and liveness properties for fully probabilistic systems.
As for the classical setting, it is established that any (probabilistic tree) property is equivalent to a conjunction of a safety and liveness property.
A simple algorithm is provided to obtain such a property decomposition for flat probabilistic CTL (PCTL). 
A safe fragment of PCTL is identified that provides a sound and complete characterisation of safety properties. 
For liveness properties, we provide two PCTL fragments, a sound and a complete one, and show that a sound and complete logical characterisation of liveness properties hinges on the (open) satisfiability problem for PCTL.
We show that safety properties only have finite counterexamples, whereas liveness properties have none.
We compare our characterisation for qualitative properties with the one for branching time properties by Manolios and Trefler, and present sound and complete PCTL fragments for characterising the notions of strong safety and absolute liveness coined by Sistla. 
\end{abstract}

\section{Introduction}

The classification of properties into safety and liveness properties is pivotal for reactive systems verification.
As Lamport introduced in 1977~\cite{Lamport1977PCM} and detailed later
in~\cite{Alford1985DSM}, safety properties assert that something
``bad'' never happens, while liveness properties require that
something ``good'' will happen eventually. 
The precise formulation of safety and liveness properties as well as their characteristics have been subject to extensive investigations.
Alpern and Schneider~\cite{alpern1987recognizing} provided a
topological characterisation in which safety properties are closed
sets, while liveness properties correspond to dense sets. 
This naturally gives rise to a decomposition---every property can be represented as a conjunction of a safety and liveness property.
It was shown that this characterisation can also be obtained using Boolean~\cite{Gumm1993GAC} and standard set theory~\cite{Rem1990personal}.
Sistla~\cite{Sistla1985CSL} studied the problem from a different
perspective and provided syntactic characterisations of safety and
liveness properties in LTL. 
The above linear-time approaches are surveyed in~\cite{Survey1994}. 
In the case of possible system failures, safety properties sometimes turn into liveness properties~\cite{DBLP:conf/concur/Charron-BostTB00}. 
The algebraic framework of Gumm~\cite{Gumm1993GAC} has been further generalised by Manolios and Trefler to characterise safety and
liveness properties both in the linear-time setting~\cite{Manolios2003LCS} as well
as in the branching-time setting~\cite{Manolios2001SLB}. 
Earlier work by Bouajjani \emph{et al.}~\cite{DBLP:conf/icalp/BouajjaniFGRS91} characterises regular safety properties by tree automata and formulas of a branching time logic. 
Alternatives to the safety-liveness taxonomy have been given in~\cite{DBLP:conf/sigsoft/NaumovichC00}.

The taxonomy of properties is not just of theoretical interest, but plays an important role in verification.
Safety and liveness properties require different proof methods~\cite{Owicki1982PLP}.
Whereas global invariants suffice for safety properties, liveness is
typically proven using proof lattices or well-founded induction and
ranking functions. 
Model checking of safety properties is usually easier than checking liveness properties~\cite{Kupferman2001MCS}.
Fairness assumptions are often imposed to exclude some unrealistic executions~\cite{Francez1986FAI}.
As fairness constraints only affect infinite computations, they can be ignored in the verification of safety properties, 
typically simplifying the verification process.
Abstraction techniques are mostly based on simulation pre-order relations that preserve safety, but no liveness properties.
Compositional techniques have been tailored to safety properties~\cite{DBLP:journals/tosem/CheungK99}.

This paper focuses on a formal characterisation of safety and liveness properties in the \emph{probabilistic} setting.
For the verification of linear-time properties, one typically resorts to using LTL or $\omega$-automata.
In the branching-time setting, mostly variants of CTL such as PCTL~\cite{Hansson94alogic} are exploited.
This is the setting that we consider.
PCTL is one of the most popular logics in the field of probabilistic model checking.
Providing a precise characterisation of safety and liveness properties for probabilistic models is highly relevant.
It is useful for identifying the appropriate analysis algorithm and provides mathematical insight.
In addition, many techniques rely on this taxonomy. 
Let us give a few examples.
Assume-guarantee frameworks~\cite{DBLP:conf/tacas/KwiatkowskaNPQ10,KomuravelliPC12} and abstraction
techniques~\cite{DBLP:conf/cav/HermannsWZ08,DBLP:journals/jlp/KatoenKLW12} aim at safety properties. 
Recent verification techniques based on monitoring~\cite{Sistla2011MSDS} indicate that arbitrary high levels
of accuracy can only be achieved for safety properties. 
Similar arguments force statistical model checking~\cite{DBLP:journals/iandc/YounesS06} to be limited to safety
properties. 
Optimal synthesis for safety properties in probabilistic games can also be done more efficiently than for liveness
properties~\cite{DBLP:conf/cav/ChatterjeeHJS10}. 

Despite the importance of distinguishing safety and liveness
properties in probabilistic systems, this subject has  (to the best of our
knowledge) not been systematically studied. 
The lack of such a framework has led to different notions of safety and
liveness properties~\cite{Baier2005CBS,Chadha2010CAF}. 
We will show that a systematic treatment leads to new insights and indicates some deficiencies of existing logical fragments for safety and liveness properties.
Inspired by~\cite{Manolios2001SLB}, we consider properties as sets of
probabilistic trees and provide a decomposition result stating that
every property can be represented by a conjunction of a safety and
liveness property. 
Moreover, all properties of the classification in the traditional
setting, such as closure of property classes under Boolean operators,
are shown to carry over to probabilistic systems. 
We study the relationship of safety and liveness properties to finite and infinite
counterexamples~\cite{Han2009CGP}, and compare our taxonomy with the classification
in~\cite{Manolios2001SLB} for qualitative properties. 
A major contribution is the identification of logical fragments of PCTL to characterise safety and liveness.
It is shown that fragments in the literature~\cite{Baier2005CBS} can be extended (for safety), or are inconsistent with our definitions (for liveness).
In addition, we consider absolute liveness and strong safety as originated by Sistla~\cite{Sistla1994SLF} for the linear-time setting.
Phrased intuitively, strong safety properties are closed under stuttering and are insensitive to the deletion of states, while once an absolutely live property holds, it is ensured it holds in the entire past.
We obtain a sound and complete characterisation of strong safety and---in contrast to \cite{Sistla1994SLF}---of absolute liveness.
In addition, we show that every absolutely live formula is equivalent to positive reachability.
This result could be employed to simplify a formula prior to verification in the same way as \cite{EtessamiH00} to 
simplify LTL formulas by rewriting in case they are stable (the complement of absolutely live) or absolutely live.
Summarising, the main contributions of this paper are:
\begin{itemize}
\item 
A formal characterisation for safety and liveness properties yielding
a decomposition theorem, i.e., every property can be represented as a conjunction of a safety and liveness property.
\item The relation of the characterisation to counterexamples.
\item
A linear-time algorithm to decompose a flat, i.e., unnested PCTL formula into a conjunction of safety and liveness properties.
\item 
A PCTL fragment that is a sound and complete characterisation of
safety properties.  (Here, completeness means that every safety
property expressible in PCTL can be expressed in the logical fragment.) 
The same applies to absolute liveness and strong safety properties.
\item 
A PCTL fragment that is a sound characterisation of liveness properties, and a fragment that is complete.
We discuss the difficulty to obtain a single sound and complete syntactic characterisation by relating it
to the PCTL decidability problem.
\item 
The relation of the property characterisation to simulation pre-orders~\cite{JonssonL91}. 
\end{itemize}

\paragraph{Organisation of the paper} 
Section~\ref{sec:pre} provides some preliminary definitions.
Section~\ref{sec:safety and liveness} presents the characterisation of safety and liveness properties.
We show the relations to counterexamples and qualitative properties of
our characterisation in Section~\ref{sec:counterexample} and
\ref{sec:qualitative properties} respectively.  
Safety PCTL is considered in Section~\ref{sec:safety pctl}, while liveness PCTL is discussed in Section~\ref{sec:liveness pctl}.
We show in Section~\ref{sec:simulation} that the new notions of safety and liveness properties
can also characterise strong simulation.
Section~\ref{sec:strong and absolute} gives the full characterisation for strong safety and absolute liveness PCTL.
Section~\ref{sec:conclusion} concludes the paper.
All proofs are included in the appendix.

\section{Preliminaries}\label{sec:pre}

For a countable set $S$, let $\powerset(S)$ denote its powerset. 
A distribution is a function $\mu:S\to [0,1]$ satisfying $\sum_{s\in S} \mu(s)= 1$. 
Let $\mathit{Dist}(S)$ denote the set of distributions over $S$. 
We shall use $s, r, t, \ldots$ and $\mu, \nu, \ldots$ to range over $S$ and $\mathit{Dist}(S)$, respectively.
The support of $\mu$ is defined by $\mathit{supp}(\mu) =\{s\in S \mid \mu(s)>0\}$. 
Let $S^*$ and $S^{\omega}$ denote the set of finite sequences and infinite sequences, respectively, over the set $S$. 
The set of all (finite and infinite) sequences over $S$ is given by $S^\infty = S^* \cup S^{\omega}$.
Let $\len{\pi}$ denote the length of $\pi \in S^\infty$ with $\len{\pi} = \infty$ if $\pi \in S^{\omega}$. 
For $i\in\mathbb{N}$, let $\pi[i]$ denote the $i{+}1$-th element of $\pi$ provided $i <\len{\pi}$, and $\lastState{\pi} \, = \pi[\len{\pi}{-}1]$ denote the last element of $\pi$ provided $\pi\in S^*$. 
A sequence $\pi_1$ is a prefix of $\pi_2$, denoted $\pi_1 \preceq \pi_2$, if $\len{\pi_1} \leq \len{\pi_2}$
and $\pi_1[i] = \pi_2[i]$ for each $0 \leq i< \len{\pi_1}$. 
Sequence $\pi_1$ is a proper prefix of $\pi_2$, denoted $\pi_1 \prec \pi_2$, if $\pi_1 \preceq \pi_2$ and  $\pi_1 \neq \pi_2$. 
The concatenation of $\pi_1$ and $\pi_2$, denoted $\pi_1 \concat \pi_2$, is the sequence obtained by appending $\pi_2$ to the end of $\pi_1$, provided $\pi_1$ is finite. 
The set $\Pi \subseteq S^\infty$ is \emph{prefix-closed} iff for all $\pi_1\in \Pi$ and $\pi_2 \in S^*$, $\pi_2 \preceq \pi_1$ implies $\pi_2 \in \Pi$.

\subsection{Discrete-Time Markov Chains}\label{sec:dtmc}
This paper focuses on discrete-time Markov chains (MCs).
Although we consider state-labelled models, all results can be transferred to action-labelled models in a straightforward way.

\begin{definition}[Markov chain]
\label{def:dtmc}
  A \emph{Markov chain} (\emph{MC}) is a tuple $\dtmc =
  (\states, \AP, \rightarrow, L, s_0)$, where $\states$ is a countable
  set of states, $\AP$ is a finite non-empty set of atomic
  propositions, $\rightarrow:\states\mapsto\dist(\states)$ is a
  transition function, $L:\states\mapsto\powerset(\AP)$ is a labelling
  function, and $s_0 \in \states$ is the initial state.
\end{definition}

\begin{figure}[!t]
  \centering
\scalebox{0.8}{
\begin{tikzpicture}[->,>=stealth,auto,node distance=3.5cm,semithick,scale=1,every node/.style={scale=1}]
	\tikzstyle{state}=[minimum size=20pt,circle,draw,thick]
	\tikzstyle{stateNframe}=[]
	every label/.style=draw
        \tikzstyle{blackdot}=[circle,fill=black, minimum size=6pt,inner sep=0pt]
     \node[state,label={[label distance=0pt]90:{$a$}}](s0){$s_0$};
     \node[state,label={[label distance=0pt]90:{$a$}}](s1)[right of=s0,yshift=1cm]{$s_1$};
     \node[state,label={[label distance=0pt]90:{$c$}}](s2)[right of=s0,yshift=-1cm]{$s_2$};
     \node[state,label={[label distance=0pt]90:{$a$}}](t0)[right of=s0, xshift=2cm]{$t_0$};
     \node[state,label={[label distance=0pt]90:{$b$}}](t1)[right of=t0,yshift=1cm]{$t_1$};
     \node[state,label={[label distance=0pt]90:{$c$}}](t2)[right of=t0,yshift=-1cm]{$t_2$};
     \node[stateNframe](a)[below of=s2,yshift=2.5cm, xshift=-1.5cm]{(a)};
     \node[stateNframe](b)[below of=t2,yshift=2.5cm, xshift=-1.5cm]{(b)};
     \path (s0) edge             node[left,yshift=5pt] {0.5} (s1)
                edge             node[left,yshift=-5pt] {0.5} (s2)
	   (s1) edge[loop right] node {1}   (s1)
	   (s2) edge[loop right] node {1}   (s2)
	   (t0) edge             node[left,yshift=5pt] {0.4} (t1)
	        edge             node[left,yshift=-5pt] {0.4} (t2)
		edge[loop below] node {0.2} (t0)
	   (t1) edge[loop right] node {1}   (t1)
	   (t2) edge[loop right] node {1}   (t2);
\end{tikzpicture}
}
  \caption{\label{fig:dtmc} Examples of MCs}
\end{figure}
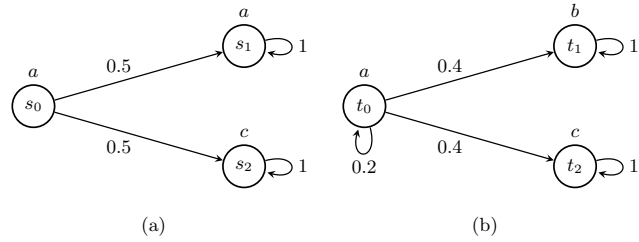
 
Fig.~\ref{fig:dtmc} presents two sample MCs where circles denote states, symbols inside the states and attached to the states denote the name and label of a state respectively.
A path $\pi\in\states^\infty$ through MC $\dtmc$ is a (finite or infinite) sequence of states. 
The cylinder set $C_\pi$ of $\pi\in\states^*$ is defined as: $C_\pi = \{\pi' \in \states^{\omega} \mid \pi \prec \pi' \}$. 
The $\sigma$-algebra $\mathcal{F}$ of $\dtmc$ is the smallest $\sigma$-algebra containing all cylinder sets $C_\pi$. 
By standard probability theory, there exists a unique probability measure $\Prob$ on $\mathcal{F}$ such that:
$\Prob(C_{\pi})=1$ if $\pi=s_0$, and $\Prob(C_{\pi}) = \Pi_{0 \leq i <n}\ \mu_i(s_{i+1})$ if $\pi=s_0\ldots s_n$ with $n>0$, where $s_i \rightarrow \mu_i$ for $0 \le i <n$.
Otherwise $\Prob(C_{\pi})=0$.  

\subsection{Probabilistic CTL}
\label{sec:pctl}
Probabilistic CTL (PCTL for short, ~\cite{Hansson94alogic}) is a branching-time logic for specifying properties of probabilistic systems. 
Its syntax is defined by the grammar:
\begin{align*}
\Phi & ::= \ a \mid \Phi_1\wedge\Phi_2\mid\neg \Phi\mid\P{\bowtie q}{\varphi}\\
\varphi & ::=\ \X \Phi\mid\Phi_1 \U \Phi_2 \mid \Phi_1\W \Phi_2
\end{align*}
where $a\in\AP$, $\bowtie\ \in\{<,>,\leq,\geq\}$ is a binary comparison operator on the reals, and $q \in [0,1]$. 
Let $\top= a\lor\neg a$ denote true and $\bot=\neg\top$ denote false.
As usual, $\Diamond\Phi = \top\U\Phi$ and $\Box\Phi=\Phi\W\bot$.
We will refer to $\Phi$ and $\varphi$ as state and path formulas, respectively.
The satisfaction relation $s \models\Phi$ for state $s$ and state formula $\Phi$ is defined in the standard manner for the Boolean connectives. 
For the probabilistic operator, it is defined by:
$
s \models\P{\bowtie q}{\varphi} \text{ iff } \Prob \{\pi\in\states^{\omega}(s)\mid\pi\models\varphi\} \bowtie q,
$
where $\states^{\omega}(s)$ denotes the set of infinite paths starting from $s$.
For MC $\dtmc$, we write $\dtmc\models\Phi$ iff its initial state satisfies $\Phi$, i.e., $s_0\models\Phi$.
The satisfaction relation for $\pi \in \states^\omega$ and path formula $\varphi$ is defined by:
\begin{align*}
&\pi\models\X\Phi &&\text{iff } \pi[1]\models\Phi\\
&\pi\models\Phi_1\U\Phi_2 && \text{iff } \exists j\geq 0.\pi[j]\models\Phi_2\land\forall 0\leq k<j.\pi[k]\models\Phi_1\\
&\pi\models\Phi_1\W\Phi_2 && \text{iff }\pi\models\Phi_1\U\Phi_2\lor\forall i\ge 0.\pi[i]\models\Phi_1.
\end{align*}
The until $\U$ and weak until $\W$ modalities are dual:
\begin{align*}
  \P{\ge q}{\Phi_1\U\Phi_2} &\equiv \P{\le 1-q}{(\Phi_1\land\neg\Phi_2)\W(\neg\Phi_1\land\neg\Phi_2)},\\
  \P{\ge q}{\Phi_1\W\Phi_2} &\equiv \P{\le 1-q}{(\Phi_1\land\neg\Phi_2)\U(\neg\Phi_1\land\neg\Phi_2)}.
\end{align*}
These duality laws follow directly from the known equivalence
$\neg (\Phi_1\U\Phi_2) \equiv (\Phi_1\wedge\neg\Phi_2)\W(\neg\Phi_1\wedge\neg\Phi_2)$ in the usual setting.
Every PCTL formula can be transformed into an equivalent PCTL formula in \emph{positive normal form}. 
A formula is in positive normal form, if negation only occurs adjacent to atomic propositions.
In the sequel, we assume PCTL formulas to be in positive normal form.

\section{Safety and Liveness Properties}
\label{sec:safety and liveness}

\subsection{Probabilistic Trees}

This section introduces the concept of probabilistic trees together with prefix and suffix relations over them.
These notions are inspired by~\cite{Manolios2001SLB}. 
Let $A,B,\ldots$ range over $\powerset(\AP)$, where $\{a\}$ is abbreviated by $a$.
Let $\epsilon$ be the empty sequence.
\begin{definition}[Probabilistic tree]\label{def:pt}
A \emph{probabilistic tree} (PT) is a tuple $\tree=\atree$ where $\epsilon\not\in\treeNodes$, and
  \begin{itemize}
  \item $(\treeNodes \cup \{\epsilon\}) \subseteq \natureNum^*$ is an unlabelled tree, i.e., prefix-closed,
  \item $L: \treeNodes\mapsto\powerset(\AP)$ is a node labelling function,
  \item $\func:\treeNodes\mapsto\mathit{Dist}(\treeNodes)$ is an edge
    labelling function, which is a partial function satisfying $\func(\pi)(\pi')>0$ iff $\pi'=\pi\concat n\in\treeNodes$ for some $n\in\natureNum$.
  \end{itemize}
\end{definition}

The node $\pi$ with $|\pi|=1$ is referred to as the \emph{root}, while all nodes $\pi$ such that $P(\pi)$ is undefined are referred to as the \emph{leaves}.
To simplify the technical presentation, $\epsilon$ is excluded from the tree. 
This will become clear after introducing the PT semantics for MCs.
PT $\tree=\atree$ is \textit{total} iff for each $\pi_1\in\treeNodes$ there exists $\pi_2\in\treeNodes$ such that $\pi_1 \prec\pi_2$, otherwise it is \textit{non-total}.
$\tree$ is \textit{finite-depth} if there exists $n\in\mathbb{N}$ such that
$\len{\pi}\le n$ for each $\pi\in\treeNodes$. Let $\alltrees^\omega$ and $\alltrees^*$ denote the sets of all total PTs and finite-depth PTs respectively,
and $\alltrees^\infty=\alltrees^*\cup\alltrees^\omega$. 
If no confusion arises, we often write a PT as a subset of $((0,1]\times\powerset(AP))^*$, i.e., as a set of sequences of its edge labelling and  node labelling functions.

\begin{example}[Probabilistic trees]
Fig.~\ref{fig:pt} depicts the finite-depth PT $\tree=\atree$.
Circles represent nodes and contain the node label and the order of the node respectively.
$$\treeNodes=\{0,00,01,02,000,001,002,011,022\}$$ and functions $L$ and $\func$ are defined in the obvious way, e.g., $L(00)=a$ and $\func(00,001)=0.4$.
PT $\tree$ can also be written as:
\begin{align*}
&\{ (1,a) ,(1,a)(0.2,a),(1,a)(0.4,b),(1,a)(0.4,c),\\ 
&\phantom{\{ } (1,a)(0.2,a)(0.2,a),(1,a)(0.2,a)(0.4,b),\\
&\phantom{\{ } (1,a)(0.2,a)(0.4,c),(1,a)(0.4,b)(1,b),\\
&\phantom{\{ } (1,a)(0.4,c)(1,c)\}.
\end{align*}
\end{example}

\begin{figure}[tbh]
  \centering
\scalebox{0.8}{
 \begin{tikzpicture}[->,>=stealth,auto,node distance=2cm,semithick,scale=1,every node/.style={scale=1}]
	\tikzstyle{state}=[minimum size=20pt,circle,draw,thick]
	\tikzstyle{stateNframe}=[]
	every label/.style=draw
        \tikzstyle{blackdot}=[circle,fill=black, minimum size=6pt,inner sep=0pt]
	\node[state](b1){$a,0$};
	\node[state](b2)[right of=b1]{$b,1$};
	\node[state](b3)[right of=b2]{$c,2$};
	\node[state](b4)[right of=b3]{$b,1$};
	\node[state](b5)[right of=b4]{$c,2$};
	\node[state](m1)[above of=b2]{$a,0$};
	\node[state](m2)[above of=b4]{$b,1$};
	\node[state](m3)[above of=b5]{$c,2$};
	\node[state](t1)[above of=m2]{$a,0$};
	\path (m1) edge            node[left] {0.2} (b1)
	           edge            node {0.4} (b2)
		   edge            node {0.4} (b3)
	      (m2) edge            node {1}   (b4)
	      (m3) edge            node {1}   (b5)
              (t1) edge            node[left,yshift=5pt] {0.2} (m1)
	           edge            node[left] {0.4} (m2)
		   edge            node {0.4} (m3);
 \end{tikzpicture}
}
  \caption{A sample probabilistic tree}\label{fig:pt}
\end{figure}
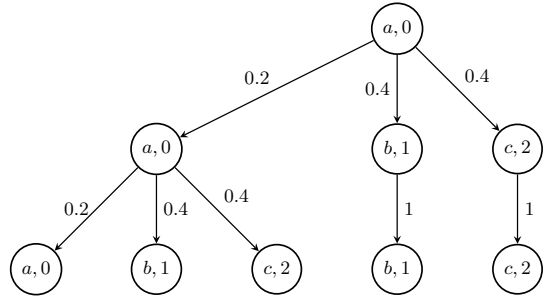

We now define when a PT is a prefix of another PT. 
\begin{definition}[Prefix]\label{def:prefix}
Let $\tree_i=\atree[i]$ for $i{=}1,2$ with $\tree_1\in\alltrees^*$ and $\tree_2\in\alltrees^\infty$.
$\tree_1$ is a \emph{prefix} of $\tree_2$, denoted $\tree_1\prefixTree\tree_2$, iff
$$
\treeNodes_1 \subseteq \treeNodes_2 \mbox{ and }
L_2 \restrict W_1 = L_1 \mbox{ and }
P_2 \restrict (W_1 \times W_1) = P_1,
$$
where $\restrict$ denotes restriction. 
Let $\finitePrefix(\tree) = \{\tree_1 \in \alltrees^* \mid \tree_1 \prefixTree \tree\}$ denote the set of all prefixes of $\tree \in \alltrees^\infty$.
\end{definition}
Conversely, we define a suffix relation between PTs:
\begin{definition}[Suffix]\label{def:suffix}
Let $\tree_i=\atree[i]$ with $\tree_i \in\alltrees^\infty$, $i=1,2$.
$\tree_2$ is a \emph{suffix} of $\tree_1$ iff there exists $\pi_1\in\treeNodes_1$ such that
 \begin{itemize}
 \item $\{\pi_1\concat\pi_2\mid\pi_2\in\treeNodes_2\}\subseteq\treeNodes_1$;
 \item $L_2(\pi_2)=L_1(\pi_1{\concat}\pi_2)$ for each $\pi_2\in\treeNodes_2$;
 \item $\func_2(\pi_2,\pi'_2)=\func_1(\pi_1{\concat}\pi_2,\pi_1{\concat}\pi'_2)$ for any $\pi_2,\pi'_2\in\treeNodes_2$.
\end{itemize}
\end{definition}
Intuitively, a suffix $\tree_2$ of $\tree_1$ can be seen as a PT obtained after executing $\tree_1$ along some sequence $\pi_1\in\treeNodes_1$.

\subsection{A PT semantics for MCs}
There is a close relation between PTs and MCs, as the execution of every MC is in fact a PT. 
Without loss of generality, we assume there exists a total order on the state space $\states$ of an MC, e.g., $\states = \mathbb{N}$.
\begin{definition}[Unfolding of an MC]\label{def:unfolding}
The \emph{unfolding} of the MC $\dtmc=(\states,\AP,\rightarrow,L,s_0)$ is the PT $\tree(\dtmc)=\atree[\dtmc]$ with: 
  \begin{itemize}
  \item $\treeNodes_\dtmc$ is the least set satisfying: i) $s_0\in\treeNodes_\dtmc$;
ii) $\pi\in\treeNodes_\dtmc$ implies $\pi \concat t \in\treeNodes_\dtmc$
for any $t\in \support(\mu)$, where $\lastState{\pi} \, \rightarrow \mu$;
  \item $L_\dtmc(\pi)=L(\lastState{\pi})$ for each $\pi\in\treeNodes_\dtmc$;
  \item $\func_\dtmc(\pi,\pi')=\mu(\lastState{\pi'})$ where $\lastState{\pi} \, \rightarrow\mu$.
  \end{itemize}
\end{definition}
Note the initial state $s_0$ is the root of the tree $\tree(\dtmc)$.
\begin{example}[Prefix, suffix and unfolding]
Let $\tree_2$ be the PT depicted in Fig.~\ref{fig:pt} and $\tree_1$ be a PT written by
$\{(1,a),(1,a)(0.2,a),(1,a)(0.4,b),(1,a)(0.4,c)\}.$ 
It follows that $\tree_1$ is a prefix of $\tree_2$.
Actually, $\tree_1$ is a fragment of $\tree_2$.
PT $\tree_1$ can be seen as a partial execution of MC $\dtmc$ in Fig.~\ref{fig:dtmc}(b) up to two steps, while $\tree_2$ is a partial execution of $\dtmc$ up to 3 steps. 
By taking the limit over the number of steps to infinity, one obtains the total PT $\tree(\dtmc)$.
Note that $\tree_1$ and $\tree_2$ are both prefixes of $\tree(\dtmc)$.

Let $\tree_3=\{(1,b),(1,b)(1,b),(1,b)(1,b)(1,b),\ldots\}$ be a total PT.
By Def.~\ref{def:suffix}, $\tree_3$ is a suffix of $\tree(\dtmc)$.
It is representing the resulting PT after jumping to $t_1$ in $\dtmc$.
\end{example}

Def.~\ref{def:unfolding} suggests to represent properties on MCs as a set of probabilistic trees.

\begin{definition}[Property] 
A \emph{property} $P\subseteq\alltrees^\omega$ is a set of total PTs. 
Property $P$ (over $\AP$) is satisfied by an MC $\dtmc$ (over $\AP$), denoted $\dtmc\models P$, iff $\tree(\dtmc)\in P$.
\end{definition}

The complement of $P$, denoted $\overline{P}$, equals $\alltrees^\omega \setminus P$.
In the sequel, let $P_\Phi = \{\tree(\dtmc)\mid\dtmc\models\Phi\}$ denote the property corresponding to the PCTL-formula $\Phi$. 
By a slight abuse of notation, we abbreviate $P_\Phi$ by $\Phi$ when it causes no confusion. 

\subsection{Safety and Liveness}
\label{sec:safetyandliveness}

Along the lines of Alpern and Schneider~\cite{alpern1987recognizing}, let us define safety and liveness properties.
\begin{definition}[Safety]\label{def:safety}
$P\subseteq\alltrees^\omega$ is a \emph{safety property} iff for all $\tree\in\alltrees^\omega$:
  $\tree\in P\text{ iff } \forall\tree_1\in\finitePrefix(\tree). \, (\exists\tree_2\in P. \, \tree_1\prefixTree\tree_2).$
\end{definition}
Thus, a safety property $P$ only consists of trees $\tree$ for which any finite-depth prefix of $\tree$ can be extended to a PT in $P$. 
Colloquially stated, if $\tree \not\in P$, there is a finite-depth prefix
of $\tree$, in which ``bad things'' have happened in finite depth and are
not irremediable. 

\begin{definition}[Liveness]\label{def:liveness}
$P\subseteq\alltrees^\omega$ is a \emph{liveness property} iff: 
$\forall\tree_1\in\alltrees^*. \, \exists\tree_2\in P. \, \tree_1\prefixTree\tree_2.$
\end{definition}
Intuitively, a property $P$ is live iff for any finite-depth PT, it is possible to extend it such that the resulting PT satisfies $P$.
Colloquially stated, it is always possible to make ``good things'' happen eventually. 
As in the classical setting, it holds that $\emptyset$ is a safety property, while $\alltrees^\omega$ is the only property which is both
safe and live.

\begin{example}[Classification of sample PCTL formulas]\label{ex:classification}\ 
\begin{itemize}
\item 
$\Phi=\P{\leq 0.5}{a\U b}$ is a safety property.\\
This can be seen as follows.
First, note that $\tree\in\Phi$ and $\tree_1\in\finitePrefix(\tree)$ implies the existence of $\tree_1\prefixTree\tree_2:=T$ and $\tree_2\in\Phi$. 
The other direction goes by contraposition.
Assume $\tree\not\in \Phi$, but for all
$\tree_1\in\finitePrefix(\tree)$, there exists $\tree_2\in\Phi$ such
that $\tree_1\prefixTree\tree_2$ (assumption *). 
If $\tree\not\in\Phi$, i.e., $\tree\in\P{>0.5}{a\U b}$, there must
exist $\tree_1\in\finitePrefix(\tree)$ in which the probability of
reaching a $b$-state via $a$-states exceeds $0.5$.  
Therefore, $\tree_1\not\prefixTree\tree_2$ for any $\tree_2\in\Phi$.
This contradicts the assumption (*).
\item 
$\Phi=\P{\geq 0.5}{a\U b}$ is neither safe nor live.\\
Let MC $\dtmc$ be depicted in Fig.~\ref{fig:dtmc}(a). 
Every finite-depth PT $\tree_1$ with $\tree_1\prefixTree\tree(\dtmc)$ can easily be extended to $\tree_2$ such that $\tree_2\in\Phi$ and $\tree_1\prefixTree\tree_2$. 
But obviously $\tree(\dtmc)\not\in\Phi$. 
Therefore $\Phi$ is not a safety property.  
To show that $\Phi$ is not a liveness property, let $\tree_1=\{(1,a),(1,a)(p,a),(1,a)(1-p,c)\}$ with $p<0.5$. 
For any possible extension of $\tree_1$, the probability of satisfying $a\U b$ is at most $p<0.5$. 
Therefore $\Phi$ is not live. 
\item 
$\Phi=\P{\geq 0.5}{\Diamond b}$, $\Phi=\P{>0.5}{\Diamond b}$ are liveness properties.\\
For every finite-depth PT $\tree_1$, there exists $\tree_2\in\Phi$ such that $\tree_1\prefixTree\tree_2$ (obtained by extending $\tree_1$ with $b$-states).
\item 
$\Phi=\P{< 0.5}{a\U b}$ is neither safe nor live.\\
Consider the MC $\dtmc$ in Fig.~\ref{fig:dtmc}(b). 
Since the probability of reaching a $b$-state $t_1$ is 0.5, $\tree(\dtmc)\not\in\Phi$. 
The probability of reaching $t_1$ in finitely many steps is however strictly less than 0.5.
Thus, for any $\tree_1\in\finitePrefix(\tree(\dtmc))$, there exists $\tree_2\in\Phi$ with $\tree_1\prefixTree\tree_2$.  
Therefore $\Phi$ is not a safety property. 
Moreover, PTs like $\tree_1=\{(1,c)\}$ show that $\Phi$ is not a liveness property either.

Remark that $\P{\le 0.5}{a\U b}$ is a safety property, whereas $\P{<0.5}{a\U b}$ is neither safe nor live. 
This can be seen as follows.
Intuitively, $\tree\not\models\P{\le 0.5}{a\U b}$ iff $\tree\models\P{>0.5}{a\U b}$, i.e., the probability of
paths in $\tree$ satisfying $a \U b$ exceeds 0.5. 
For this, there must exist a set of \emph{finite} paths in $\tree$ satisfying $a\U b$ whose probability mass exceeds 0.5. 
However, this does not hold for $\P{<0.5}{a \U b}$, as $\tree\not\models\P{< 0.5}{a\U b}$ iff $\tree\models\P{\ge 0.5}{a\U b}$. 
There exist PTs (like the one in Fig.~\ref{fig:dtmc}(b)) such that they satisfy $\P{\ge 0.5}{a\U b}$, but the probability mass of their \emph{finite} paths satisfying $a\U b$ never exceeds 0.5.
\item 
$\Phi=\P{> 0.4}{a\U b}$ is neither safe nor live.\\ 
Consider the MC $\dtmc$ in Fig.~\ref{fig:dtmc}(a).
Clearly, $\dtmc \not\models \Phi$, as the probability of reaching a $b$-state is 0. 
But any finite-depth prefix of $\tree(\dtmc)$ can be extended to a PT in $\Phi$. 
Thus, $\Phi$ is not a safety property.
Moreover for finite-depth PTs like $\tree_1=\{(1,c)\}$, there exists no $\tree_2\in\Phi$ such that $\tree_1\prefixTree\tree_2$. 
Therefore $\Phi$ is not a liveness property.
\end{itemize}
\end{example}

\subsection{Characterisations of Safety and Liveness}

As a next step, we aim to give alternative characterisations of safety and liveness properties using topological closures~\cite{Manolios2003LCS}.
\begin{definition}[Topological closure]\label{def:topological closure}
Let $X$ be a set.
The function $\tco:\powerset(X)\mapsto\powerset(X)$ is a \emph{topological closure} operator on a $X$ iff for any $C,D\subseteq X$ it holds:
  \begin{enumerate}
  \item $\tco(\emptyset)=\emptyset$;
  \item $C\subseteq\tco(C)$;
  \item $\tco(C)=\tco(\tco(C))$;
  \item $\tco(C\cup D)=\tco(C)\cup\tco(D)$.
  \end{enumerate}
\end{definition}

The following lemma shows two important properties of topological closure operators, where $\comp{C} = X\setminus C$ denotes the complement of $C$ w.r.t.\, $X$.
\begin{lemma}[\cite{Manolios2003LCS}]\label{lem:topological closure}
For a topological closure operator $\tco$ on $X$ and
$C\subseteq X$ we have:
  \begin{itemize}
  \item $\tco(C\cup\comp{\tco(C)})=X$;
  \item $\tco(C)\cap(C\cup\comp{\tco(C)})=C$.
  \end{itemize}
\end{lemma}

A closure function maps sets of total trees onto sets of total trees.
It is in particular useful when applied to properties.
\begin{definition}[Property closure]\label{def:closure linear}
Let $\closure: \powerset(\alltrees^\omega)\rightarrow\powerset(\alltrees^\omega)$.
The \emph{closure} of property $P \subseteq \alltrees^\omega$ is defined by:
$$
\closure(P) = \{\tree\in\alltrees^\omega\mid\forall\tree_1\in\finitePrefix(\tree).(\exists\tree_2\in P.\tree_1\prefixTree\tree_2)\}.
$$
\end{definition}
Intuitively speaking, $\closure(P)$ is the set of probabilistic trees for which all prefixes have an extension in $P$.
Consider the topological space $(\alltrees^\omega, \powerset(\alltrees^\omega))$. 
It follows:
\begin{lemma}\label{lem:topological closure}
The function $\closure$ is a topological closure operator on $(\alltrees^\omega, \powerset(\alltrees^\omega))$.
\end{lemma}

The following theorem provides a topological characterisation of safety and liveness
for probabilistic systems, which can be seen as a conservative extension of the results in~\cite{Manolios2003LCS}.
\begin{theorem}\label{thm:safety and liveness characterisation}
\ 
  \begin{enumerate}
  \item $P$ is a safety property iff $P=\closure(P)$.
  \item $P$ is a liveness property iff $\closure(P)=\alltrees^\omega$.
  \end{enumerate}
\end{theorem}

Theorem~\ref{thm:safety and liveness characterisation} asserts that a property is safe iff its closure coincides with itself. 
A property $P$ is live iff the closure of $P$ equals $\alltrees^\omega$, i.e., the set of all total PTs.

\begin{remark}\label{rm:decomposition}
From these results, it follows that $P\cup\comp{\closure(P)}$ is a liveness property for any $P$.
Using Lemma~\ref{lem:topological closure}, we have
$\closure(P\cup\comp{\closure(P)}) =
\closure(P)\cup\closure(\comp{\closure(P)})\supseteq\closure(P)\cup\comp{\closure(P)}=\alltrees^\omega$.  
Therefore $\closure(P\cup\comp{\closure(P)})=\alltrees^\omega$.
By Theorem~\ref{thm:safety and liveness characterisation}, it follows that $P\cup\comp{\closure(P)}$ is 
a liveness property.
\end{remark}

Theorem~\ref{thm:safety and liveness characterisation} and
Remark~\ref{rm:decomposition} provide the basis for a decomposition
result stating that every property can be represented as an
intersection of a safety and liveness property. 
\begin{proposition}[Decomposition proposition]\label{prop:decompose}
For any property $P\subseteq\alltrees^\omega$, $P=\closure(P)\cap(P\cup\comp{\closure(P)})$.
\end{proposition}
We thus can decompose any property $P$ into the intersection of the
properties $\closure(P)$ and $(P\cup\comp{\closure(P)})$, where
$\closure(P)$ is a safety property by Theorem~\ref{thm:safety and
  liveness characterisation}, and $P\cup\comp{\closure(P)}$ is a
liveness property by Remark~\ref{rm:decomposition}. 
Finally, we study whether safety and liveness properties are closed under conjunction and disjunction.
\begin{lemma}\label{lem:disjunction and conjunction}
Given two properties $P_1$ and $P_2$:
  \begin{enumerate}
  \item Safety properties are closed under $\cap$ and $\cup$;
  \item If $P_1$ and $P_2$ are live with $P_1\cap P_2\neq\emptyset$, so is $P_1\cap P_2$;
  \item If at least one of $P_1$ and $P_2$ is live, so is $P_1\cup P_2$.
  \end{enumerate}
\end{lemma}
Lemma~\ref{lem:disjunction and conjunction} provides a means to prove safety and liveness properties
in a compositional way. For instance, in order to prove that $P_1\cap P_2$ is safe,
we can prove whether $P_1$ and $P_2$ are safe or not separately. In case that both $P_1$ and $P_2$
are safe, so is $P_1\cap P_2$.


\subsection{Safety and liveness versus counterexamples}
\label{sec:counterexample}

We conclude this section by providing a relationship between safety and liveness properties and counterexamples.
A property $P$ only has finite counterexamples iff for any MC
$\dtmc\not\models P$, there exists
$\tree_1\in\finitePrefix(\tree(\dtmc))$ with
$\tree_1\not\prefixTree\tree_2$ for any $\tree_2\in P$.  
Conversely, a property $P$ has no finite counterexamples iff for any
MC $\dtmc$ such that $\dtmc\not\models P$, for each
$\tree_1\in\finitePrefix(\tree(\dtmc))$ there exists $\tree_2\in P$
such that $\tree_1\prefixTree\tree_2$, i.e., no finite-depth prefix is
able to violate the property. 
\begin{theorem}\label{thm:counterexample safety and liveness}
\ 
\begin{enumerate}
\item  $P$ is safe iff it only
  has finite counterexamples.
\item  $P$ is live iff
  it has no finite counterexamples.
\end{enumerate}
\end{theorem}


Recall that $\Phi=\P{\leq 0.5}{a\U b}$ is a safety property.  
As shown in~\cite{Han2009CGP}, for any MC $\dtmc\not\models\Phi$, there exists a (finite) set of finite paths of $\dtmc$ whose mass probability exceeds 0.5. 
This indicates that $\Phi$ only has finite counterexamples.

\section{Qualitative Properties}
\label{sec:qualitative properties}
\begin{table}
\centering
\caption{Property classification of qualitative PCTL}\label{tab:qualitative}
\begin{tabular}{|c|c|c|c|c|c|}
\hline
\multicolumn{2}{|c|}{Qualitative PCTL} & \multirow{2}{*}{Equivalence}  &
\multicolumn{3}{|c|}{CTL}\\
\cline{1-2}\cline{4-6}formula & here & & formula & \cite{Manolios2001SLB} & \cite{alpern1987recognizing} \\
\hline \hline 
$\P{=1}{\Diamond a}$ & L & $\not\equiv$ & $\forall\Diamond a$ & UL & L\\[1ex]
$\P{>0}{\Diamond a}$ & L & $\equiv$ & $\exists\Diamond a$ & EL & L
\\[1ex]
$\P{>0}{a\U b}$ & X & $\equiv$ & $\exists(a\U b)$ & X & X\\[1ex]
$\P{=1}{\Box a}$ & S & $\equiv$ & $\forall\Box a$ & US & S\\[1ex]
$\P{>0}{\Box a}$ & X & $\not\equiv$ & $\exists\Box a$ & ES & S\\[1ex]
\hline
\end{tabular}
\end{table}

The qualitative fragment of PCTL only contains formulas with probability bounds $\geq 1$ (or $=1$) and $>0$.
Although CTL and qualitative PCTL have incomparable expressive power~\cite{Baier2008PMC}, they have a large fragment in common.
(For finite MCs, qualitative PCTL coincides with CTL under strong fairness assumptions.)
This provides a basis for comparing the property classification
defined above to the existing classification for branching-time
properties~\cite{Manolios2001SLB}. 
A qualitative PCTL-formula $\Phi$ is equivalent to a CTL-formula
$\Psi$ whenever $\dtmc \models \Phi$ iff $\dtmc \models \Psi$, where
the latter is interpreted over the underlying digraph of MC $\dtmc$. 

\begin{example}[Classifying qualitative PCTL versus CTL/LTL] \
\begin{itemize}
\item 
$\P{=1}{\Diamond a}$ and $\forall\Diamond a$.  
Although $\P{=1}{\Diamond a} \not\equiv \forall\Diamond a$, both formulas are liveness properties. 
Recall that $\P{=1}{\Diamond a} \equiv\P{\ge 1}{\top\U a}$, which is a liveness property (see Example~\ref{ex:classification}).
\item 
$\P{>0}{\Diamond a}$ and $\exists\Diamond a$.
As $\P{>0}{\Diamond a}\equiv\P{>0}{\top\U a}$ it follows from Example~\ref{ex:classification} that $\P{>0}{\Diamond a}$ is a liveness property. 
According to~\cite{Manolios2001SLB}, \emph{CTL}-formula $\exists\Diamond a$ is a universally liveness property. 
Note that $\forall\Diamond a$ and $\exists\Diamond a$ coincide in the linear-time setting of~\cite{alpern1987recognizing}.
\item 
$\P{>0}{a\U b}$ and $\exists(a\U b)$.
Note $\P{>0}{a\U b}\equiv\exists(a\U b)$. 
In fact, also their classifications coincide: the \emph{PCTL}-formula $\P{>0}{a\U b}$ is neither safe nor live (see
Example~\ref{ex:classification}), whereas the CTL-formula $\exists(a\U b)$ is also neither safe nor live~\cite{Manolios2001SLB}. 
Similarly, in the linear-time setting, $a \U b$ is neither safe nor live~\cite{alpern1987recognizing}. 
\item 
$\P{=1}{\Box a}$ and $\forall\Box a$.
In this case, $\P{=1}{\Box a}\equiv\forall\Box a$ (see~\cite{Baier2008PMC}). 
Since $\P{=1}{\Box a} \equiv \P{\leq 0}{a \U\neg a}$, it follows from
Example~\ref{ex:classification} that $\P{=1}{\Box a}$ is safe. 
This coincides with the characterisation of $\forall\Box a$ in~\cite{alpern1987recognizing}.
\item 
$\P{>0}{\Box a}$ and $\exists\Box a$.
As shown in~\cite{Baier2008PMC}, $\P{>0}{\Box a}\not\equiv\exists\Box a$.
This non-equivalence is also reflected in the property characterisation. 
Since $\P{>0}{\Box a}\equiv \P{<1}{a\U\neg a}$, it is neither safe nor live (see Example~\ref{ex:classification}).
In contrast, $\exists\Box a$ is classified as a safety property and
existentially safety property in~\cite{alpern1987recognizing}
and~\cite{Manolios2001SLB}, respectively. 
\end{itemize}
\end{example}

Table~\ref{tab:qualitative} summarises the classification where L, S,
and X denote liveness, safety, and other properties respectively,
while the prefixes E and U denote \textit{existentially} and
\textit{universally} 
respectively.  
The second column indicates our characterisation, while the 5th and
6th column present the characterisation of~\cite{Manolios2001SLB}
and~\cite{alpern1987recognizing} respectively. 
Please bear in mind, that~\cite{alpern1987recognizing} considers linear-time properties. 

In conclusion, our characterisation for qualitative PCTL coincides
with that of~\cite{alpern1987recognizing} and \cite{Manolios2001SLB} with the exception of
$\P{>0}{\Box a}$.  
\cite{Manolios2001SLB} considers the branching-time setting, and
treats two types of safety properties: universally safety (such as
$\forall\Box a$) and existentially safety (e.g., $\exists\Box a$).  
The same applies to liveness properties.
Accordingly, \cite{Manolios2001SLB} considers two closure operators:
one using finite-depth prefixes (as in Def.~\ref{def:closure linear})
and one taking non-total prefixes into account. 
The former is used for universally safety and liveness
properties, the latter for existentially safety and
liveness. This explains the mismatches in Table~\ref{tab:qualitative}. 
We remark that our characterisation of qualitative properties will
coincide with~\cite{Manolios2001SLB} by using a variant of $\closure$
that considers non-total prefixes.

\section{Safety PCTL}\label{sec:safety pctl}

In this section, we will provide syntactic characterisations of safety properties in PCTL.
For flat PCTL, in which nesting is prohibited, we present an algorithm to decompose a flat PCTL-formula into a conjunction of a safe and live formula.
Then we provide a sound and complete characterisation for full PCTL.
In both setting, formulas with strict probability bounds are excluded.

\subsection{Flat PCTL}\label{sec:flat pctl}
Here we focus on a flat fragment of PCTL, denoted $\pctlf$, whose syntax is given by the following grammar:
$$
\Phi ::= \P{\bowtie q}{\Phi^a_1\U\Phi^a_2} \mid \P{\bowtie q}{\Phi^a_1\W\Phi^a_2} \mid \P{\bowtie q}{\X\Phi^a} \mid\Phi_1 \wedge \Phi_2 \mid \Phi_1 \vee \Phi_2
$$
with $\bowtie \, \in \{\le,\ge\}$, and $\Phi^a ::= a \mid \neg\Phi^a \mid \Phi^a_1\wedge\Phi^a_2$ is referred to as \emph{literal formulas}. 
The fragment $\pctlf$ excludes nested probabilistic operators as well as strict probability bounds.
Note that by applying the distribution rules of disjunction and conjunction, 
every formula $\Phi$ in $\pctlf$ can be transformed into an equivalent
formula such that all conjunctions are at the outermost level except for those between literal formulas $\Phi^a$. 
Therefore we assume all $\pctlf$-formulas to obey such form.
We provide an algorithm that decomposes a $\pctlf$-formula into a conjunction of two PCTL-formulas, 
one of which is a safety property, while the other one is a liveness property.
$\pctlf$ is closed under taking the closure:
\begin{lemma}\label{lem:closure PCTL}
The closure formula of a $\pctlf$-formula equals:
$$
\begin{array}{rcl}
\closure(\Phi^a) & = & \Phi^a \\
\closure(\P{\bowtie q}{\X\Phi^a}) & = & {\P{\bowtie q}{\X\Phi^a}} \mbox{ for } \bowtie \, \in \{ \leq, \geq \} \\
\closure(\P{\le q}{\Phi^a_1\U\Phi^a_2}) & = & {\P{\le q}{\Phi^a_1\U\Phi^a_2}} \\
\closure(\P{\ge q}{\Phi^a_1\U\Phi^a_2}) & = & {\P{\geq q}{\Phi^a_1\W\Phi^a_2}} \\
\closure(\P{\ge q}{\Phi^a_1\W\Phi^a_2}) & = & {\P{\ge q}{\Phi^a_1\W\Phi^a_2}} \\
\closure(\P{\le q}{\Phi^a_1\W\Phi^a_2}) & = & {\P{\le q}{\Phi^a_1\U\Phi^a_2}} \\
\closure({\Phi_1\lor\Phi_2}) & = & \closure({\Phi_1})\lor\closure({\Phi_2}).
\end{array}
$$
%
\end{lemma}
By Lemma~\ref{lem:closure PCTL}, the size of $\closure(\Phi)$ is linear in the size of
$\Phi$ for any $\pctlf$ formula $\Phi$.
In Lemma~\ref{lem:closure PCTL}, we do not define the closure formula
for conjunctions, as in general it does not hold that
$\closure(\Phi_1\land\Phi_2)=\closure(\Phi_1)\land\closure(\Phi_2)$: 
\begin{example}[Closure of conjunctions]
Let $\Phi=\Phi_1\land\Phi_2$ where $\Phi_1=\P{\ge 1}{a\U b}$ and $\Phi_2=\P{\ge 1}{(a\land\neg b)\U(\neg a\land\neg b)}$.
It follows that $\Phi\equiv\bot$. We show that $\closure(\Phi)\neq\closure(\Phi_1)\land\closure(\Phi_2)=\P{\ge 1}{a\W b}\land\P{\ge 1}{(a\land\neg b)\W(\neg a\land\neg b)}$.
Since a PT always staying in $a$-states almost surely is in $\closure(\Phi_1)\land\closure(\Phi_2)$,
$\closure(\Phi_1)\land\closure(\Phi_2)\not\equiv\bot$. However $\closure(\Phi)\equiv\bot$ because $\Phi\equiv\bot$.
\end{example}

Algorithm~\ref{alg:safety and liveness} describes the procedure of decomposition.
It is worth mentioning that given $\Phi\in\pctlf$, Algorithm~\ref{alg:safety and liveness} returns a pair of formulas $(\Phi^s,\Phi^l)$
such that $\Phi\equiv\Phi^s\land\Phi^l$, where $\Phi^s\in\pctlf$, but $\Phi^l$ is not necessary in $\pctlf$.

\begin{algorithm}[!t]
    \caption{$\pctlf$ decomposition}\label{alg:safety and liveness}
    \begin{algorithmic}[1]
    \REQUIRE A $\pctlf$-formula $\Phi$.
    \ENSURE~~\\
        $(\Phi^s,\Phi^l)$ such that $\Phi^s\land\Phi^l\equiv\Phi$ where $\Phi^s$ is a safety property and $\Phi^l$ is a liveness property. \\[1ex]
        \STATE Transform $\Phi$ into an equivalent formula such that $\Phi\equiv\Phi_1\land\Phi_2\land\ldots\land\Phi_n$ where $\Phi_i$ ($1\le i\le n$)
        contains no conjunction operators except between literal formulas;\label{line:distributivity}
        \STATE Let $\Phi^s_i = \closure(\Phi_i)$ for each $1\le i\le n$ (see Lemma~\ref{lem:closure PCTL});\label{line:safety}
        \STATE Let $\Phi^l_i = \Phi_i\lor\neg\Phi^s_i$ for each $1\le i\le n$;\label{line:liveness}
        \STATE Return ($\bigwedge_{1\le i\le n}\Phi^s_i, \bigwedge_{1\le i\le n}\Phi^l_i$).\label{line:conjunction}    
\end{algorithmic}
\end{algorithm}

\begin{theorem}\label{thm:correctness of algorithm}
Algorithm~\ref{alg:safety and liveness} is correct.  
\end{theorem}

Since line 1 in Algorithm~\ref{alg:safety and liveness} may cause an exponential blow-up by transforming $\Phi$ into an equivalent formula
in conjunctive normal form. 
It follows that Algorithm~\ref{alg:safety and liveness} has an  exponential worst-case time complexity.
 
The reason for not considering formulas with strict bounds can be seen in the following example:
\begin{example}[Strict bounds]\label{ex:non-strict}
Let $\Phi=\P{>0.5}{a\U b}$. 
We show that $\closure(\Phi)$ cannot be represented in \emph{PCTL}. 
Let $\dtmc_1$ be the MC in Fig.~\ref{fig:dtmc}(b). 
Every finite-depth prefix $\tree_1$ of $\tree(\dtmc_1)$ can easily be
extended to a PT $\tree_2\in\Phi$ such that $\tree_1
\prefixTree\tree_2$.  
From Def.~\ref{def:closure linear} it follows $\tree(\dtmc_1)\in\closure(\Phi)$. 
Now consider MC $\dtmc_2$ in Fig.~\ref{fig:dtmc}(a) where we label state $s_1$ with $b$ (rather than $c$). 
Then $\tree(\dtmc_2)\not\in\closure(\Phi)$.
For instance, the finite-depth prefix $\{(1,a),(1,a)(0.5,b),(1,a)(0.5,c)\}$ of $\tree(\dtmc_2)$ 
cannot be extended to a PT in $\Phi$ as the probability of reaching $b$-states via only $a$-states is at most $0.5$. 
Applying \cite[Th.\ 50]{Baier2005CBS}, no {PCTL} $\X$-free formula can distinguish $\dtmc_1$ and $\dtmc_2$, 
as they are \emph{weakly bisimilar} (which is easy to verify).
  
The above arguments indicate that all PTs in which $\neg(a\lor
b)$-states are reached with probability $\ge$ 0.5 in finitely many
steps are not in $\closure(\Phi)$, while PTs where $\neg(a\lor
b)$-states can only be reached with probability $\ge$ 0.5 in
infinitely many steps are in $\closure(\Phi)$.   
However, in order to characterise PTs where $\neg(a\lor b)$-states can
only be reached with probability $\ge$ 0.5 in infinitely many steps,
we need infinitary conjunction of $\X$ operators.   
This is not possible in {PCTL}.
Thus, $\closure(\Phi)$ cannot be represented in {PCTL}.
\end{example}

\subsection{Safety PCTL with Nesting}
\label{sec:embedded PCTL}
In this section we aim to give a sound and complete characterisation of safety properties in PCTL.
That is to say, we will define a fragment of PCTL, that in contrast to
$\pctlf$, contains nesting of probability operators, such that each
formula in that fragment is a safety property. 
We also show the opposite, namely, that every safety property
expressible in PCTL can be expressed as a formula in the provided
logical fragment. 
For the same reasons as explained in Example~\ref{ex:non-strict},
strict probability bounds are excluded. 
The logical fragment is defined as follows.

\begin{definition}[Safety PCTL]\label{def:safty pctl}
Let $\mathcal{F}=\pctls$ denote the \emph{safe fragment} of \emph{PCTL}, defined as the smallest set satisfying:
\begin{enumerate}
\item $\Phi^a\in\mathcal{F}$;
\item If $\Phi\in\mathcal{F}$, then $\P{\ge q}{\X\Phi}\in\mathcal{F}$;
\item If $\Phi_1,\Phi_2\in\mathcal{F}$, then $\Phi_1\land\Phi_2,\Phi_1\lor\Phi_2,\P{\ge q}{\Phi_1\W\Phi_2}\in\mathcal{F}$;
\item If $\neg\Phi_1,\neg\Phi_2\in\mathcal{F}$, then $\P{\le q}{\Phi_1\U\Phi_2}\in\mathcal{F}$.
\end{enumerate}  
\end{definition}

The next result asserts that all properties in $\pctls$ are indeed safety properties according to Def.~\ref{def:safety}.
\begin{theorem}\label{thm:safety pctl}
Every $\pctls$-formula is a safety property.
\end{theorem}

The following theorem asserts (in some sense) the converse of
Theorem~\ref{thm:safety pctl}, i.e., all safety properties in PCTL can
be represented by an equivalent formula in $\pctls$. 
\begin{theorem}\label{thm:safety pctl complete}
For every safety property $\Phi$ expressible in \emph{PCTL} (no strict
bounds), there exists $\Phi'\in\pctls$ with $\Phi\equiv\Phi'$. 
\end{theorem}

Note for any $\Phi\in\pctlf$, $\closure(\Phi)\in\pctlf\cap\pctls$. 
Thus, Algorithm~\ref{alg:safety and liveness} decomposes
$\pctlf$-formula $\Phi$ into a conjunction of a safety and liveness
property such that the safety property is expressed in
$\pctlf\cap\pctls$. 

\section{Liveness PCTL}\label{sec:liveness pctl}
In this section we investigate expressing liveness properties in PCTL.
We start with providing a sound characterisation of liveness
properties, that is to say, we provide a logical fragment for liveness
properties. Subsequently, we show that a slight superset of this fragment yields a
complete characterisation of liveness properties expressible in PCTL. 
We then discuss the reasons why, in contrast to safety properties, a
syntactic sound and complete characterisation of PCTL-expressible
liveness properties is difficult to achieve. 
Let us first define the logical fragment $\pctll^{<}$.
\begin{definition}[Liveness PCTL]\label{def:liveness pctl}
  Let $\mathcal{F}=\pctll^{<}$ denote the \emph{live fragment} of \emph{PCTL}, defined as the smallest set satisfying: 
\begin{enumerate}
\item $\top\in\mathcal{F}$ and $\bot\not\in\mathcal{F}$;
\item $\P{\ge q}{\Diamond\Phi^a}\in\mathcal{F}$;
\item If $\Phi_1,\Phi_2\in\mathcal{F}$, then $\Phi_1\land\Phi_2\in\mathcal{F}$;
\item If $\Phi_1\in\mathcal{F}$ or $\Phi_2\in\mathcal{F}$, then 
$\Phi_1\lor\Phi_2,\P{\ge q}{\Phi_1\W\Phi_2}\in\mathcal{F}$;
\item If $\Phi\in\mathcal{F}$, then $\P{\ge q}{\X\Phi}\in\mathcal{F}$;
\item If $\Phi_2\in\mathcal{F}$, then $\P{\ge q}{\Phi_1\U\Phi_2}\in\mathcal{F}$ for any $\Phi_1$.\label{item:liveness U}
\end{enumerate}
\end{definition}

It follows that $\pctll^{<}$-formulas are liveness properties.
\begin{theorem}\label{thm:liveness pctl sound}
Every $\pctll^{<}$-formula is a liveness property.
\end{theorem}

However, the converse direction is not true, i.e., it is not the case
that every liveness property expressible in PCTL can be expressed in
$\pctll^{<}$. 
This is exemplified below. 

\begin{example}[A liveness property not in $\pctll^{<}$]\label{ex:liveness not complete}
Let $\Phi=\P{\ge 1}{\P{\ge 1}{\Diamond a}\U b}$. 
First, observe $\Phi\not\in\pctll^{<}$, since
$b\not\in\pctll^{<}$ according to Def.~\ref{def:liveness pctl}.
On the other hand, it follows that $\Phi$ is a liveness property. 
This can be seen as follows.
Let $\tree_1\in\alltrees^*$ be an arbitrary finite-depth PT. 
By Def.~\ref{def:safety}, it suffices to show that $\tree_1\prefixTree\tree_2$ for some $\tree_2\in\Phi$. 
Such $\tree_2$ can be constructed by extending all leaves in $\tree_1$ with a transition to $(a\land b)$-states with probability 1.
This yields $\tree_2\in\Phi$.
Therefore such $\tree_2\in\Phi$ with $\tree_1\prefixTree\tree_2$ always exists and $\Phi$ is a liveness property.
\end{example}

Example~\ref{ex:liveness not complete} shows that $\pctll^{<}$ is not
complete, i.e., it does not contain all liveness properties
expressible in PCTL.  
The problem is caused by clause~\ref{item:liveness U}) in
Def.~\ref{def:liveness pctl}, where we require that $\Phi_2 \in
\pctll^{<}$, in order for $\P{\ge q}{\Phi_1\U\Phi_2} \in \pctll^{<}$.  
As shown in Example~\ref{ex:liveness not complete},
this requirement is too strict, since it excludes liveness properties like $\P{\ge 1}{\P{\ge 1}{\Diamond a}\U b}$.
Let us now slightly relax the definition of $\pctll^{<}$ by replacing clause~\ref{item:liveness U}) in Def.~\ref{def:liveness pctl} by:
\begin{equation}\label{eq:liveness U}
\begin{aligned}
&\text{If }\Phi_1\in\mathcal{F}\text{ or }\Phi_2\in\mathcal{F}, \text{then }\P{\ge q}{\Phi_1\U\Phi_2}\in\mathcal{F}.
\end{aligned}
\end{equation}
The resulting logical fragment is referred to as $\pctll^{>}$.
This fragment contains all liveness properties expressible in PCTL.
\begin{theorem}\label{thm:liveness pctl complete}
  For any liveness property $\Phi$ expressible in \emph{PCTL}, there exists $\Phi'\in\pctll^{>}$ with $\Phi\equiv\Phi'$.
\end{theorem}
$\pctll^{>}$ is a superset of $\pctll^{<}$ and contains all liveness PCTL properties. 
Unfortunately, it also contains some properties which are not live, i.e., it is not sound.
In the example below we show that formulas like $\Phi=\P{\ge
  0.5}{\Phi_1\U\Phi_2}$ cannot be classified easily when $\Phi_1$ is a
liveness property while $\Phi_2$ is not (A live formula with a similar schema is given in
Example~\ref{ex:liveness not complete}).
\begin{example}[Liveness is hard to capture syntactically]\label{ex:liveness not sound involved}
Let $\Phi=\P{\ge 0.5}{\Phi_1\U\Phi_2}$ with $\Phi_1=\P{\ge 1}{\Diamond
  a}\land\P{\ge 1}{\Diamond(\neg a\land\neg b)}$ and $\Phi_2=\P{\ge
  1}{\Box(\neg a\land b)}$. 
Intuitively, $\Phi_1$ requires that $a$-states and $(\neg a\land \neg
b)$-states are each eventually reached almost surely, while $\Phi_2$
requires to  almost surely stay in $(\neg a\land b)$-states. 
By Def.~\ref{def:liveness pctl}, $\Phi_1\in\pctll^{<}$, which implies $\Phi_1\in\pctll^{>}$ and $\Phi\in\pctll^{>}$.
$\Phi$ is however not a liveness property.
We show this by arguing that $\tree_1=\{(1,a)\}$ is not a prefix of any PT in $\Phi$. 
Let $\tree_1\prefixTree\tree_2$.
As $\tree_2\not\in\Phi_2$, $\tree_1$ needs to be extended so as to yield a PT in $\Phi_1$ so as to fulfil $\Phi$.
Since $\Phi_1\land\Phi_2\equiv\bot$ and $a\land(\neg a\land\neg
b)\equiv\bot$, for any $\tree\in\Phi_1$, it follows
$\tree\not\in\Phi_2$ and $\tree\not\in\P{>0}{\X\Phi_2}$. 
$\Phi_1$ thus implies $\neg\Phi$.
Thus $\Phi$ is not live.  

Actually, $\Phi\equiv\Phi_2$, since it is not possible to reach $\Phi_2$-states via only $\Phi_1$-states. 
In order for a PT satisfying $\Phi$, it must satisfy $\Phi_2$ initially. 
Every $\Phi$ can be simplified to an equivalent property not in
$\pctll^{>}$.
\end{example}

In conclusion, formulas like $\Phi=\P{\ge 0.5}{\Phi_1\U\Phi_2}$ are live, provided $\Phi_2$ is live too. 
The difficulty arises when $\Phi_2$ is not live but $\Phi_1$ is. 
Since Examples~\ref{ex:liveness not complete} and \ref{ex:liveness not sound involved} indicate that the liveness of $\Phi_1$ does not
necessarily imply the liveness of $\Phi$.
Whereas the definition of safe PCTL formulas can be done inductively over the structure of the formula, this is not applicable to live PCTL.
For instance, formulas like $\P{\ge 0.5}{\Phi_1\U\Phi_2}$ cannot be categorised as being live (or not) based on the sub-formulas.

It is worth mentioning that membership in $\pctls$ can be determined
syntactically, while this does neither hold for $\pctll^{<}$ nor for
$\pctll^{>}$. 
Since, first of all, we require that $\Phi\not\equiv\bot$ for each $\Phi\in\pctll^{<}$ and $\Phi\in\pctll^{>}$. 
The checking of $\Phi\not\equiv\bot$ relies on PCTL satisfiability
checking, i.e., $\Phi\not\equiv\bot$ if and only if there exists
$\tree\in\alltrees^\omega$ such that $\tree\in\Phi$ ($\Phi$ is
satisfiable).  
PCTL satisfiability has received scant attention, and only partial
solutions are known: \cite{Brazdil2008SPP} considers satisfiability
checking for qualitative PCTL, while~\cite{BertrandFS12} presents an
algorithm for bounded satisfiability checking of bounded PCTL. 
To the best of our knowledge, no algorithm for full PCTL satisfiability checking exists.
Secondly, as indicated in Example~\ref{ex:liveness not sound
  involved}, formulas of the form $\P{\ge q}{\Phi_1\U\Phi_2}$ cannot
be easily classified syntactically. 
In order for $\pctll^{>}$ to solely contain liveness properties,
the condition Eq.~(\ref{eq:liveness U}) should be changed to: $\P{\ge
  q}{\Phi_1\U\Phi_2}\in\mathcal{F}$ iff 
\begin{enumerate}
\item either $\Phi_2\in\mathcal{F}$,
\item or $\Phi_1\in\mathcal{F}$ and $\Phi_1\land\P{\ge q}{\Phi_1\U\Phi_2}\not\equiv\bot$.
\end{enumerate}
The first clause subsumes $\pctll^{<}$, while the second clause requires
that in case only $\Phi_1$ is in $\pctll^{>}$, $\Phi_1\land\P{\ge
  q}{\Phi_1\U\Phi_2}$ must be satisfiable, namely, it is possible to
extend a PT satisfying $\Phi_1$ such that it satisfies $\P{\ge
  q}{\Phi_1\U\Phi_2}$. 

It is not surprising to encounter such difficulties when characterising PCTL liveness.
Even in the non-probabilistic setting, the characterisation of
liveness LTL relies on LTL satisfiability checking and it is (to our
knowledge) still an open problem to provide a both sound and complete
characterisation for liveness in LTL~\cite{Sistla1994SLF} and CTL. 

\begin{remark}
In contrast to Section~\ref{sec:embedded PCTL}, where safety
properties are restricted to non-strict bounds, both $\pctll^{<}$ and
$\pctll^{>}$ can be extended to strict bounds while preserving all
theorems of this section. 
\end{remark}

\section{Characterisation of Simulation Pre-order}\label{sec:simulation}

Simulation is an important pre-order relation for comparing the behaviour of MCs~\cite{JonssonL91}.
Roughly speaking, an MC $\dtmc$ simulates $\dtmc'$ whenever it can mimic all transitions of $\dtmc'$ with at least the same probability.
A logical characterisation of (weak and strong) simulation pre-order relations on MCs has been given in~\cite{Baier2005CBS}.
Baier \emph{et al.}~\cite{Baier2005CBS} use the following safety and liveness fragments of PCTL.
The safety fragment is given by:
\begin{equation}\label{eq:safety}
\Phi ::= a \mid \neg a \mid \Phi_1\land\Phi_2 \mid \Phi_1\lor\Phi_2 \mid \P{\ge p}{\X\Phi} \mid \P{\ge q}{\Phi_1\W\Phi_2},
\end{equation}
while the liveness fragment is defined by:
\begin{equation}\label{eq:liveness}
\Phi ::= a \mid \neg a \mid \Phi_1\land\Phi_2 \mid \Phi_1\lor\Phi_2 \mid \P{\ge p}{\X\Phi} \mid \P{\ge q}{\Phi_1\U\Phi_2}.
\end{equation}
Observe that $\pctls$ subsumes the safety PCTL defined in Eq.~(\ref{eq:safety}).  
In addition, formulas of the form $\P{\le q}{\Phi_1\U\Phi_2}$ belong to $\pctls$, provided $\neg\Phi_1$ and $\neg\Phi_2$ are safety properties. 
The main difference between \cite{Baier2005CBS} and our characterisation is concerned with liveness properties. 
The liveness fragment in Eq.~(\ref{eq:liveness}) is incomparable with both $\pctll^{<}$ and $\pctll^{>}$.
For instance, formulas like $\P{\ge q}{a\U b}$ are live according to
Eq.~(\ref{eq:liveness}), but is neither safe nor live according to our
characterisation. 

Now we demonstrate whether the logical fragment $\pctls$ characterises strong simulations, and similar for the two liveness fragments defined before.
The concept of strong simulation between probabilistic models relies on the concept of \emph{weight function}~\cite{Jones1989PPE,JonssonL91}:
\begin{definition}[Weight function]\label{def:weight function}
Let $\states$ be a set and $R\subseteq \states\times \states$.
A \emph{weight function} for distributions $\mu_1$ and $\mu_2$ with
respect to $R$ is a function $\Delta:\states\times \states\mapsto[0,1]$
satisfying:
\begin{itemize}
\item $\Delta(s_1,s_2)>0$ implies $s_1~R~s_2$,
\item $\mu_1(s_1)=\sum_{s_2\in\states}\Delta(s_1,s_2)$ for any $s_1\in \states$,
\item $\mu_2(s_2)=\sum_{s_1\in\states}\Delta(s_1,s_2)$ for any $s_2\in \states$.
\end{itemize}
We write $\mu_1~\weightFunc~\mu_2$ if there exists a weight function
$\Delta$ for $\mu_1$ and $\mu_2$ with respect to $R$. 
\end{definition}

Strong simulation for MCs is now defined as follows.
\begin{definition}[Strong simulation]\label{def:simulation}
Let $\dtmc=(\states, \AP, \rightarrow, L, s_0)$ be an MC.
$R \subseteq \states \times \states$ is a \emph{strong simulation} iff $s_1~R~s_2$ implies
$L(s_1)=L(s_2)$ and $\mu_1~\weightFunc~\mu_2$, where
$s_i\rightarrow\mu_i$ with $i\in\{1,2\}$.
We write $s_1~\precsim~s_2$ iff there exists a strong simulation
$R$ such that $s_1~R~s_2$.
\end{definition}

In order to give a logical characterisation of $\precsim$ using $\pctls$, we define a pre-order relation on $\pctls$.
Let $s_1~\epctls~s_2$ iff $s_2\models\Phi$ 
implies $s_1\models\Phi$ for every $\Phi\in\pctls$. 
Similarly, $s_1~\epctll^i~s_2$ iff $s_1\models\Phi$ implies $s_2\models\Phi$ for any $\Phi\in\pctll^i$ with $i\in\{1,2\}$. 
The following theorem shows that both $\epctls$ and $\epctlltwo$ can
be used to characterise strong simulation as in~\cite{Baier2005CBS},
while $\epctllone$ is strictly coarser than $\precsim$. 
\begin{theorem}\label{thm:simulation logical}
 $\precsim~=~\epctls~=~\epctlltwo~\subsetneq~\epctllone$.
\end{theorem}
The proof of $\epctlltwo~\subseteq~\precsim$ relies on liveness properties expressible in PCTL.
Consequently, $\precsim~=~\epctll$, where $\epctll$ is the pre-order
induced by $\pctll$, i.e., the set of all liveness properties
expressible in PCTL. 

\section{Strong Safety and Absolute Liveness}\label{sec:strong and absolute}

In this section, we characterise strong safety and absolute liveness properties as originated in~\cite{Sistla1985CSL} for LTL. 
In the original setting, a strong safety property $P$ is a safety
property that is closed under stuttering, and is insensitive to the
deletion of states, i.e., deleting an arbitrary number of states from
a sequence in $P$ yields a sequence in $P$. 
(A similar notion also appeared in~\cite{DBLP:journals/ipl/AlpernDS86}.)
We lift this notion to probabilistic trees and provide a sound and
complete characterisation of strong safety (expressible in PCTL). 
In contrast, an absolute liveness property is a liveness property that is insensitive to adding prefixes.
We provide a sound and complete characterisation of absolute liveness
properties, and show that each such property is in fact an almost sure
reachability formula. 

\subsection{Strong Safety Properties}

\begin{definition}[Stuttering]\label{def:stuttering}
PT $\tree_1=\atree[1]$ is a \emph{stuttering} of PT $\tree_2=\atree[2]$ iff for some $\pi_1$ with $\lastState{\pi_1} \, =n$:
$$
\treeNodes_1\setminus\treeNodes_2=\{\pi_1{\concat}n{\concat}\pi_2\mid\pi_1{\concat}\pi_2\in\treeNodes_2\}\text{, and}
$$
\begin{itemize}
\item for any $\pi\in\treeNodes_1$, 
$$
L_1(\pi) = \left\{ \begin{array}{ll}
  L_2(\pi) & \mbox{if } \pi \in \treeNodes_2 \\
  L_2(\pi_1) & \mbox{if } \pi=\pi_1{\concat}n \\
  L_2(\pi_1{\concat}\pi_2) & \mbox{if } \pi=\pi_1{\concat}n{\concat}\pi_2 \\
\end{array} \right.
$$
\item 
for any $\pi,\pi'\in\treeNodes_1$, $\func_1(\pi)(\pi')$ equals
$$
 \left\{ \begin{array}{ll}
  \func_2(\pi)(\pi') & \mbox{if } \pi,\pi' \in W_2 \\
  1 & \mbox{if } \pi=\pi_1, \pi'=\pi_1{\concat}n \\
  \func_2(\pi_1{\concat}\pi_2)(\pi_1{\concat}\pi'_2) & \mbox{if } \pi=\pi_1{\concat}n{\concat}\pi_2, \pi'=\pi_1{\concat}n{\concat}\pi'_2.
\end{array} \right.
$$
\end{itemize}
\end{definition}
Phrased in words, $\tree_1$ is the same as $\tree_2$ except that one
or more nodes in $\tree_2$, such as the last node of $\pi_1$ is
repeated (stuttered) with probability one for all paths in
$\treeNodes_1$ with prefix $\pi_1$.  
Conversely, we can also delete nodes from a PT:
\begin{definition}[Shrinking]\label{def:shrinking}
Let $\tree_1,\tree_2\in\alltrees^\omega$.
PT $\tree_1=\atree[1]$ is a \emph{shrinking} of $\tree_2=\atree[2]$
iff there exists $\pi_1{\concat} n\in\treeNodes_2$ with
$\pi_1\neq\epsilon$ such that 
$$
\treeNodes_1\setminus\treeNodes_2=\{\pi_1{\concat}\pi_2\mid\pi_1{\concat}n {\concat}\pi_2\in\treeNodes_2\}\text{, and}
$$
\begin{itemize}
\item 
for any $\pi\in\treeNodes_1$, 
$$
L_1(\pi) = \left\{ \begin{array}{ll}
 L_2(\pi) & \mbox{if } \pi\in\treeNodes_2 \\
 L_2(\pi_1{\concat}n{\concat}\pi_2) & \mbox{if } \pi=\pi_1{\concat}\pi_2.
\end{array} \right.
$$
\item 
for any $\pi,\pi'\in\treeNodes_1$, $\func_1(\pi)(\pi')$ equals
$$
\!\!\!\!\!\!\!\!\!\!\!\!
\left\{ \begin{array}{ll}
   \func_2(\pi)(\pi') & \mbox{if } \pi,\pi' \in \treeNodes_2 \\
   \func_2(\pi)(\pi_1{\concat}n){\times}\func_2(\pi_1{\concat}n)(\pi_1{\concat}n{\concat}\pi'_2)
        & \mbox{if } \pi=\pi_1, \pi'=\pi_1{\concat}\pi'_2 \\
   \func_2(\pi_1{\concat}n{\concat}\pi_2)(\pi_1{\concat}n{\concat}\pi'_2)
        & \mbox{if } \pi=\pi_1{\concat}\pi_2 \mbox{ and} \\
        & \phantom{\mbox{ifif}} \pi'=\pi_1{\concat}\pi'_2.
    \end{array} \right.
$$
\end{itemize}
\end{definition}
Note that deletion of the initial node is prohibited, as $\pi_1\neq\epsilon$.

\begin{example}[Shrinking and stuttering]\label{ex:stuttering and shrinking}
Let $\tree_1$, $\tree_2$, and $\tree_3$ be the PTs depicted in Fig.~\ref{fig:stuttering and shrinking},
where symbols inside circles denote node labels.
$\tree_2$ is a stuttering PT of $\tree_1$, as in $\tree_2$ the $c$-node is stuttered with probability one.
On the other hand, $\tree_3$ is obtained by deleting the $b$-state from $\tree_1$, such that the probability
from $a$-state to $d$-state and $e$-state equals $0.5{\times}0.4 = 0.2$ and $0.5{\times}0.6=0.3$, respectively. 
Thus, $\tree_3$ is a shrinking PT of $\tree_1$.
\end{example}

\begin{figure}[t]
  \centering
 \scalebox{1}{
 \begin{tikzpicture}[->,>=stealth,auto,node distance=2cm,semithick,scale=0.8,every node/.style={scale=0.8}]
	\tikzstyle{state}=[minimum size=20pt,circle,draw]
	\tikzstyle{stateNframe}=[]
	every label/.style=draw
        \node[stateNframe](d11)                           {$\vdots$};
        \node[state](d1)[above of=d11]                    {$d$};
        \node[state](e1)[right of=d1]                     {$e$};
        \node[stateNframe](d12)[below of=e1]              {$\vdots$};
        \node[state](b1)[above of=d1,xshift=1cm]          {$b$};
        \node[state](a1)[above of=b1,xshift=1cm]          {$a$};
        \node[state](c1)[right of=b1]                     {$c$};
        \node[stateNframe](d13)[below of=c1]              {$\vdots$};
        \node[stateNframe](d21)[right of=d13,yshift=-2cm,xshift=-1cm] {$\vdots$};
        \node[state](d2)[above of=d21]                    {$d$};
        \node[state](e2)[right of=d2]                     {$e$};
        \node[stateNframe](d22)[below of=e2]              {$\vdots$};
        \node[state](b2)[above of=d2,xshift=1cm]          {$b$};
        \node[state](a2)[above of=b2,xshift=1cm]          {$a$};
        \node[state](c2)[right of=b2]                     {$c$};
        \node[state](c21)[below of=c2]                    {$c$};
        \node[stateNframe](d23)[below of=c21]             {$\vdots$};
        \node[stateNframe](d31)[right of=c21,xshift=-1cm]             {$\vdots$};
        \node[stateNframe](d32)[right of=d31,xshift=-1cm]             {$\vdots$};
        \node[stateNframe](d33)[right of=d32,xshift=-1cm]             {$\vdots$};
        \node[state](d3)[above of=d31]                    {$d$};
        \node[state](e3)[above of=d32]                    {$e$};
        \node[state](c3)[above of=d33]                    {$c$};
        \node[state](a3)[above of=e3]                     {$a$};
        \node[stateNframe](n1)[below of=d12,yshift=1cm]   {PT $\tree_1$};
        \node[stateNframe](n2)[below of=d22,yshift=1cm]   {PT $\tree_2$};
        \node[stateNframe](n3)[below of=d32,yshift=-1cm]   {PT $\tree_3$};
	\path 
              (a1)  edge            node [left] {0.5} (b1)
                    edge            node [right]{0.5} (c1)
              (b1)  edge            node [left] {0.4} (d1)
                    edge            node [right]{0.6} (e1)
              (c1)  edge            node        {1}   (d13)
              (d1)  edge            node [left] {1}   (d11)
              (e1)  edge            node [right]{1}   (d12)
              (a2)  edge            node [left] {0.5} (b2)
                    edge            node [right]{0.5} (c2)
              (b2)  edge            node [left] {0.4} (d2)
                    edge            node [right]{0.6} (e2)
              (c2)  edge            node        {1}   (c21)
              (c21) edge            node        {1}   (d23)
              (d2)  edge            node [left] {1}   (d21)
              (e2)  edge            node [right]{1}   (d22)
              (a3)  edge            node [left] {0.2} (d3)
                    edge            node [xshift=-0.1cm,yshift=-0.2cm]  {0.3} (e3)
                    edge            node [right]{0.5} (c3)
              (d3)  edge            node        {1}   (d31)
              (e3)  edge            node        {1}   (d32)
              (c3)  edge            node        {1}   (d33);
 \end{tikzpicture}
 }
  \caption{Illustrating stuttering and shrinking of PTs}\label{fig:stuttering and shrinking}
\end{figure}
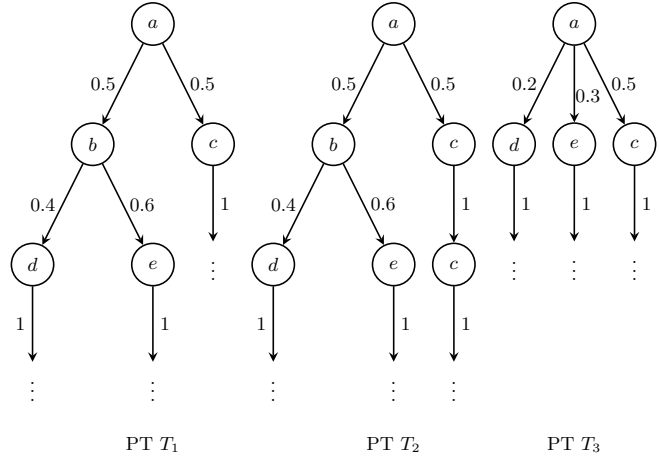

Now we are ready to define the \emph{strong safety} properties in the probabilistic setting:
\begin{definition}[Strong safety]\label{def:strong safety}
A safety property $P$ is a \emph{strong safety} property whenever
\begin{enumerate}
\item 
$P$ is closed under stuttering, i.e, $\tree\in P$ implies $\tree'\in P$,
for every stuttering PT $\tree'$ of $\tree$, and
\item 
$P$ is closed under shrinking, i.e., $\tree\in P$ implies $\tree'\in P$,
for every shrinking PT $\tree'$ of $\tree$.
\end{enumerate}
\end{definition}
Observe that there exist non-safety properties that are closed under stuttering and shrinking. 
For instance $\P{\ge 0.5}{\top\U\P{\ge 1}{\Box a}}$ is not a safety property, but is closed under stuttering and shrinking.
In~\cite{Sistla1994SLF}, it was shown that an LTL formula is a strong
safety property iff it can be represented by an LTL formula in
positive normal form using only $\Box$ operators. 
We extend this result in the probabilistic setting: strong safety properties syntactically
cover more PCTL-formulas than those only containing $\Box$ operators. 

\begin{definition}[Strong safety PCTL]\label{def:strong safety pctl}
Let $\mathcal{F}=\pctlss$ denote the \emph{strong safety fragment} of $\pctls$ such that:
\begin{enumerate}
\item $\Phi^a\in\mathcal{F}$;
\item If $\Phi_1,\Phi_2\in\mathcal{F}$, then
$\Phi_1\land\Phi_2$ and $\Phi_1\lor\Phi_2$ are in $\mathcal{F}$;
\item\label{item:ss U} If $\Phi_1\in\mathcal{F}$ and $\Phi_2\in\mathcal{F}^\Box$, then
  $\P{\ge q}{\Phi_1\W\Phi_2}\in\mathcal{F}$;
\end{enumerate}
where $\mathcal{F}^\Box$ is defined as follows:
\begin{enumerate}
\item If $\Phi_1,\Phi_2\in\mathcal{F}^\Box$, then $\Phi_1\land\Phi_2$ and $\Phi_1\lor\Phi_2$ are in $\mathcal{F}^\Box$;
\item If $\Phi\in\mathcal{F}$, then $\P{\ge 1}{\Box\Phi}\in\mathcal{F}^\Box$.
\end{enumerate}
\end{definition}
Note that by clause~\ref{item:ss U}), $\P{\ge q}{\Box\Phi}$ is a formula in $\pctlss$, provided $\Phi\in\pctlss$.
This follows from the fact that $\P{\ge q}{\Box\Phi}\equiv\P{\ge
  q}{\Phi\W\bot}\equiv\P{\ge q}{\Phi\W\P{\ge 1}{\Box\bot}}$, and
$\P{\ge 1}{\Box\bot}\in\mathcal{F}^\Box$. 
The following result shows that $\pctlss$ is sound and complete, i.e.,
all formulas in $\pctlss$ are strong safety properties and every
strong safety property expressible in PCTL is expressible in
$\pctlss$. 
\begin{theorem}\label{thm:strong sound and complete}
Every $\pctlss$-formula is a strong safety property and for any strong
safety property $\Phi$ expressible in \emph{PCTL}, there exists
$\Phi'\in\pctlss$ with $\Phi\equiv\Phi'$. 
\end{theorem}
The question whether all formulas in $\pctlss$ can be represented by
an equivalent formula in positive normal form using only
$\Box$-modalities is left for future work.

\subsection{Absolute Liveness Properties} 
Now we introduce the concepts of \emph{stable} properties and \emph{absolute liveness} properties.
Intuitively, a property $P$ is stable, if for any $\tree\in P$, all suffixes of $\tree$ are also in $P$.
This intuitively corresponds to once $P$ is satisfied, it will never be broken in the future.
\begin{definition}[Stable property]\label{def:stable}
$P$ is a \emph{stable property} iff $\tree\in P$ implies $\tree'\in P$,
for every suffix $\tree'$ of $\tree$.
\end{definition}
A property $P$ is an absolute liveness property, if for any $\tree\in P$, all PTs which have $\tree$ as a suffix are also in $P$.
Colloquially stated,  once $P$ is satisfied at some point, $P$ was satisfied throughout the entire past. 
\begin{definition}[Absolute liveness]\label{def:absolute liveness}
  $P$ is an \emph{absolute liveness} property iff $P\neq\emptyset$ and 
  $\tree'\in P$ implies $\tree\in P$, for every suffix $\tree'$ of $\tree$.
\end{definition}

Rather than requiring every absolutely liveness property to be a liveness property by definition, this follows implicitly:
\begin{lemma}\label{lem:absolute is liveness}
Every absolute liveness property is live.
\end{lemma}

For transition systems, there is a close relationship between stable and absolute liveness properties~\cite{Sistla1994SLF}. 
A similar result is obtained in the probabilistic setting:
\begin{lemma}\label{lem:stable and absolute}
  For any $P\neq\alltrees^\omega$, $P$ is a stable property iff $\comp{P}$ is
  an absolute liveness property.
\end{lemma}

\begin{definition}[Absolute liveness PCTL]\label{def:absolute liveness pctl}
Let $\mathcal{F}=\pctlal$ denote the \emph{absolute liveness fragment} of \emph{PCTL} such that:
\begin{enumerate}
\item $\top\in\mathcal{F}$ and $\bot\not\in\mathcal{F}$;
\item If $\Phi_1,\Phi_2\in\mathcal{F}$, then
$\Phi_1\land\Phi_2$, $\Phi_1\lor\Phi_2$, $\P{>0}{\Phi_1\W\Phi_2}\in\mathcal{F}$;
\item If $\Phi_2\in\mathcal{F}$, then $\P{> 0}{\X\Phi_2},\P{>0}{\Phi_1\U\Phi_2}\in\mathcal{F}$;
\item\label{item:al} If $\Phi_1\in\mathcal{F}$ with $\neg\Phi_1\land\Phi_2\equiv\bot$,
  then $\P{> 0}{\Phi_1\U\Phi_2},\P{>0}{\Phi_1\W\Phi_2}\in\mathcal{F}$.
\end{enumerate}
\end{definition}

According to the definition of $\pctlal$, $\pctlal$ only contains qualitative properties
with bound $>0$. 
By clause~\ref{item:al}), $\P{>0}{\Diamond\Phi}$ is
an absolute liveness formula for any $\Phi\not\equiv\bot$, while
$\P{>0}{\Box\Phi}$ is an absolute liveness formula provided $\Phi$ is
so too. Note that $\pctlal$ is a proper subset of $\pctll^{>}$ but not of $\pctll^{<}$, e.g., 
formulas like $\P{>0}{\Phi_1\U\Phi_2}$ with $\Phi_1=\P{>0}{\Diamond b}$ and $\Phi_2=\P{\ge 0.5}{a\U b}$ is in $\pctlal$
because $\Phi_1\in\pctlal$ and $\neg\Phi_1\land\Phi_2\equiv\bot$. 
However $\Phi\not\in\pctll^{<}$, since $\Phi_2\not\in\pctll^{<}$.
\begin{theorem}\label{thm:absolute sound and complete}
Every formula in $\pctlal$ is an absolute liveness property, and for
every absolute liveness property $\Phi$ expressible in \emph{PCTL}, there
exists $\Phi'\in\pctlal$ with $\Phi\equiv\Phi'$. 
\end{theorem}
Inspired by~\cite{Sistla1994SLF}, we provide an alternative characterisation of absolute liveness properties.
\begin{theorem}\label{thm:absolute characterisation}
\emph{PCTL}-formula $\Phi$ is an absolute liveness property iff
$\Phi\not\equiv\bot$ and $\Phi\equiv\P{>0}{\Diamond\Phi}$.
\end{theorem}

\section{Conclusions}
\label{sec:conclusion}

This paper presented a characterisation of safety and liveness properties for fully probabilistic systems.
It was shown that most facts from the traditional linear-time~\cite{alpern1987recognizing} and branching-time setting~\cite{Manolios2003LCS} are preserved.
In particular, every property is equivalent to the conjunction of a safety and liveness property.
Various sound PCTL-fragments have been identified for safety, absolute liveness, strong safety, and liveness properties.
Except for liveness properties, these logical characterisation are all complete.
Fig.~\ref{fig:summary} summarises the PCTL-fragments and their relation, where $L_1\rightarrow L_2$ denotes that $L_2$ is a sub-logic of $L_1$.\footnote{
Here, it is assumed that $\pctll^{<}$ and $\pctll^{>}$ also support strict bounds.
}
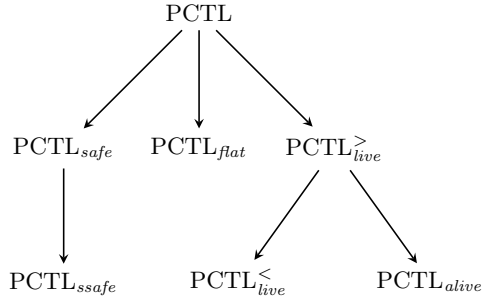
\begin{figure}[h]
  \centering
\scalebox{1}{ 
\begin{tikzpicture}[->,>=stealth,auto,node distance=1.8cm,semithick]
	\tikzstyle{state}=[]
	\tikzstyle{stateNframe}=[]
	every label/.style=draw
        \node[state](ss){$\pctlss$};
        \node[state](safe)[above of=ss]{$\pctls$};
        \node[state](flat)[right of=safe]{$\pctlf$};
        \node[state](live2)[right of=flat]{$\pctll^{>}$};
        \node[state](pctl)[above of=flat]{PCTL};
        \node[state](live1)[below left of=live2,yshift=-0.5cm]{$\pctll^{<}$};
        \node[state](alive)[below right of=live2,yshift=-0.5cm]{$\pctlal$};
	\path (pctl)  edge            node {}   (safe)
                      edge            node {}   (flat)
                      edge            node {}   (live2)
	      (safe)  edge            node {}   (ss)
	      (live2) edge            node {}   (live1)
                      edge            node {}   (alive);
 \end{tikzpicture}
}
\caption{Overview of relationships between PCTL fragments}\label{fig:summary}
\end{figure}

There are several directions for future work such as extending the
characterisation to Markov decision processes, considering
fairness~\cite{Volzer2005DF}, finite executions~\cite{Maier04}, and
more expressive logics such as the probabilistic
$\mu$-calculus~\cite{DBLP:journals/corr/abs-1211-1511}. 


\newpage
\appendix
\section{Proofs}

\begin{replemma}{lem:topological closure}
The function $\closure$ is a topological closure operator
  on $(\alltrees^\omega, \powerset(\alltrees^\omega))$.
\end{replemma}
\begin{proof}
  We show that $\closure$ satisfies the four properties in
  Def.~\ref{def:topological closure}.
  \begin{enumerate}
  \item $\closure(\emptyset)=\emptyset$.
   This case is straightforward from Def.~\ref{def:closure linear}.
  \item $P\subseteq\closure(P)$. 
   We show that for each $\tree\in P$, 
   $\tree\in\closure(P)$. According to Def.~\ref{def:closure
     linear}, $\tree\in\closure(P)$ iff for each
   $\tree_1\in\finitePrefix(\tree)$, there
   exists $\tree_2\in P$ such that
   $\tree_1\prefixTree\tree_2$. By choosing
   $\tree$ as $\tree_2$, we obtain $\tree\in\closure(P)$.
 \item $\closure(P)=\closure(\closure(P))$.
   The previous case indicates that
   $\closure(P)\subseteq\closure(\closure(P))$, so we only need to
   show $\closure(\closure(P))\subseteq\closure(P)$. Suppose that
   $\tree\in\closure(\closure(P))$. By Def.~\ref{def:closure linear}, 
   for each $\tree_1\in\finitePrefix(\tree)$, there exists
   $\tree_2\in\closure(P)$ such that
   $\tree_1\prefixTree\tree_2$, i.e., $\tree_1\in\finitePrefix(\tree_2)$. 
   As $\tree_2\in\closure(P)$, there exists $\tree'_2\in P$
   such that $\tree_1\prefixTree\tree'_2$. Thus $\tree\in\closure(P)$.
 \item$\closure(P\cup P')=\closure(P)\cup\closure(P')$. The proof of
   $\closure(P)\cup\closure(P')\subseteq\closure(P\cup P')$ is
   straightforward from Def.~\ref{def:closure linear}. For the other
   direction, let $\tree\in\closure(P\cup P')$. We proceed
   by contraposition and assume $\tree\not\in\closure(P)$ and
   $\tree\not\in\closure(P')$.  First, (i) $\tree\not\in\closure(P)$
   implies there exists $\tree_1\in\finitePrefix(\tree)$ such that
   there does not exist $\tree_1'\in P$ satisfying
   $\tree_1\prefixTree\tree_1'$. Similarly, (ii)
   $\tree\not\in\closure(P')$ implies there exists
   $\tree_2\in\finitePrefix(\tree)$ such that there does not exist
   $\tree_2'\in P'$ satisfying $\tree_2\prefixTree\tree_2'$. Let
   $\tree_3\in\finitePrefix(\tree)$ such that
   $\tree_1\prefixTree\tree_3$ and $\tree_2\prefixTree\tree_3$.
   Since $\tree_1$ and $\tree_2$ are finite-depth prefixes of $\tree$, such
   $\tree_3$ always exists. By construction, (i) implies that there
   does not exist $\tree_1'\in P$ such that
   $\tree_3\prefixTree\tree_1'$, and (ii) implies that there does not
   exist and $\tree_2'\in P'$ such that
   $\tree_3\prefixTree\tree_2'$. This implies that
   $\tree\not\in\closure(P\cup P')$, a contradiction.
 \end{enumerate}
  \vspace*{-0.5cm}
\end{proof}

\begin{reptheorem}{thm:safety and liveness characterisation}
\  
 \begin{enumerate}
  \item $P$ is a safety property iff $P=\closure(P)$.
  \item $P$ is a liveness property iff $\closure(P)=\alltrees^\omega$.
  \end{enumerate}
\end{reptheorem}
\begin{proof} \
\begin{enumerate}
\item
\begin{itemize}
\item[$\Rightarrow$] Let $P$ be a safety property.   Lemma~\ref{lem:topological closure}  implies $P\subseteq\closure(P)$. For the other direction let $\tree\in\closure(P)$. According to
  Def.~\ref{def:closure linear}, for all
  $\tree_1\in\finitePrefix(\tree)$, there
  exists $\tree_2\in P$ such that
  $\tree_1\prefixTree\tree_2$. By Def.~\ref{def:safety} it follows $\tree\in P$.
\item[$\Leftarrow$]Let $P=\closure(P)$ and
  $T\in\alltrees^\omega$. First assume $T\in P$ and let
  $\tree_1\in\finitePrefix(\tree)$. Then there exists $\tree_2:=T\in
  P$ with $\tree_1\prefixTree\tree_2$.  Moreover, assume that for all
  $\tree_1\in\finitePrefix(\tree)$, there exists $\tree_2\in P$ with
  $\tree_1\prefixTree\tree_2$.  Def.~\ref{def:closure linear} implies
  that $T\in\closure(P)=P$. Therefore, $P$ is a safety property.
\end{itemize}  
\item
\begin{itemize}
\item[$\Rightarrow$]
Assume $P$ is a liveness
  property.  Obviously it holds
  $\closure(P)\subseteq\alltrees^\omega$. For the other direction let
  $\tree\in\alltrees^\omega$. Fix arbitrary
  $\tree_1\in\finitePrefix(\tree)$. Since $P$ is a liveness property,  there exists $\tree_2\in P$ such that $\tree_1\prefixTree\tree_2$. Thus  $\tree\in\closure(P)$.
  \item[$\Leftarrow$] Assume $\closure(P)=\alltrees^\omega$. By contraposition.
  Suppose that $P$ is not a liveness property. By Def.~\ref{def:liveness},
  there exists $\tree_1\in\alltrees^*$ such that $\tree_1\not\prefixTree\tree_2$ for
  all $\tree_2\in P$. Let $\tree$ be a tree such that $\tree_1\in\finitePrefix(\tree)$. Then we have
  $\tree\not\in\closure(P)$ according to Def.~\ref{def:closure linear}. This
  contradicts the assumption that $\closure(P)=\alltrees^\omega$.
\end{itemize}
\end{enumerate} 
\vspace*{-0.5cm}
\end{proof}

\begin{replemma}{lem:disjunction and conjunction}
Given two properties $P_1$ and $P_2$:
  \begin{enumerate}
  \item Safety properties are closed under $\cap$ and $\cup$;
  \item If $P_1$ and $P_2$ are live with $P_1\cap P_2\neq\emptyset$, so is $P_1\cap P_2$;
  \item If at least one of $P_1$ and $P_2$ is a liveness property, so is $P_1\cup P_2$.
  \end{enumerate}
\end{replemma}
\begin{proof} \ 
  \begin{enumerate}
  \item
    Let $P_1$ and $P_2$ be safety properties.  According to
    Def.~\ref{def:topological closure} and Lemma~\ref{lem:topological
      closure}, $\closure(P_1\cup
    P_2)=\closure(P_1)\cup\closure(P_2)=P_1\cup P_2$. Therefore by
    Theorem~\ref{thm:safety and liveness characterisation}, $P_1\cup
    P_2$ is a safety property. We now prove that $P_1\cap P_2$ is
    also a safety property. Clearly $P_1\cap
    P_2\subseteq\closure(P_1\cap P_2)$. We prove that
    $\closure(P_1\cap P_2)\subseteq P_1\cap P_2$. Let
    $\tree\in\closure(P_1\cap P_2)$. Thus for arbitrary
    $\tree_1\in\finitePrefix(\tree)$, there exists $\tree_2\in P_1\cap
    P_2$ with $\tree_1\prefixTree\tree_2$. Obviously $\tree_2\in P_1$ and $\tree_2\in P_2$. 
    Since both $P_1$ and $P_2$ are safety properties, we have $\tree\in P_1$
    and $\tree\in P_2$, i.e., $\tree\in P_1\cap P_2$.
     \item 
    Let $P_1$ and $P_2$ be two liveness properties. We prove that $P_1\cap
    P_2$ is also a liveness property, provided that $P_1\cap
    P_2\neq\bot$. Let $\tree_1\in\alltrees^*$ be an arbitrary finite-depth PT,
    and let $\tree\in P_1\cap P_2$. We construct
    $\tree_2\in\alltrees^\omega$ by appending $\tree$ at all leaves of
    $\tree_1$. Since $P_1$ and $P_2$ are liveness properties, 
    $\tree_2\in P_1$ and $\tree_2\in P_2$.  Therefore
    $\tree_2\in P_1\cap P_2$ as desired.

  \item Suppose $P_1$ is a liveness property, i.e., $\closure(P_1)=\alltrees^\omega$. 
    According to Def.~\ref{def:topological closure} and Lemma~\ref{lem:topological closure}, $\closure(P_1\cup P_2)=\closure(P_1)\cup\closure(P_2)=\alltrees^\omega$, 
    therefore $P_1\cup P_2$ is a liveness property.
  \end{enumerate}
  \vspace*{-0.5cm}
\end{proof}


\begin{reptheorem}{thm:counterexample safety and liveness}\ 
\begin{enumerate}
\item  $P$ is safe iff it only
  has finite counterexamples.
\item  $P$ is live iff
  it has no finite counterexamples.
\end{enumerate}
\end{reptheorem}
\begin{proof}\ 
\begin{enumerate}
\item $P$ is safe iff it only has finite counterexamples.
  \begin{enumerate}
  \item[$\Rightarrow$]We first prove that if $P$ is a safety property, then it 
    has only finite counterexamples.
    By contraposition.
    Assume $P$ is a safety property and there exists an MC $\dtmc$ such that $\dtmc\not\models P$,
    but for all $\tree_1\in\finitePrefix(\tree(\dtmc))$, there exists
    $\tree_2\in P$ such that $\tree_1\prefixTree\tree_2$. 
    This indicates that $\tree(\dtmc)\in P$, since $P$ is a safety
    property, which contradicts the assumption that $\dtmc\not\models
    P$. 
  \item[$\Leftarrow$] Secondly, we prove that if for every 
    $\dtmc\not\models P$, there exists $\tree_1\in\finitePrefix(\tree(\dtmc))$ such that
    $\tree_1\not\prefixTree\tree_2$ for any
    $\tree_2\in P$, then $P$ is a safety property. 
    Again we proceed by contraposition. Assume $P$ is not a safety
    property. According to Def.~\ref{def:safety}, there exists
    $\tree\not\in P$ such that for all $\tree_1\in\finitePrefix(\tree)$, there exists
    $\tree_2\in P$ such that $\tree_1\prefixTree\tree_2$. 
    Let $\dtmc$ be an MC with $\tree(\dtmc)=\tree$, then
    $\dtmc\not\models P$, but there does not exist
    $\tree_1\in\finitePrefix(\tree(\dtmc))$ such that
    $\tree_1\not\prefixTree\tree_2$ for all
    $\tree_2\in P$. Contradiction.
  \end{enumerate}
\item $P$ is live iff it has no finite counterexamples.
  \begin{enumerate}
  \item[$\Rightarrow$] Given a liveness property $P$, we show that for any
    MC $\dtmc$ such that $\dtmc\not\models P$, it has no finite
    counterexamples. Suppose that there exists $\tree_1\in\finitePrefix(\tree(\dtmc))$ such that $\tree_1\not\prefixTree\tree_2$
    for any $\tree_2\in P$, this contradicts with the fact that $P$ is a
    liveness property.
  \item[$\Leftarrow$] Suppose that for any MC $\dtmc\not\models P$ and $\tree_1\in\finitePrefix(\tree(\dtmc))$, 
    there exists $\tree_2\in P$ such that
    $\tree_1\prefixTree\tree_2$. By contraposition, if $P$ is not a liveness property,
    then there exists a $\tree_1\in\alltrees^*$ such that
    $\tree_1\not\prefixTree\tree_2$ for all
    $\tree_2\in P$. Let $\dtmc$ be an MC such that
    $\tree_1\prefixTree\tree(\dtmc)$, then $\dtmc\not\models P$, but the
    finite-depth prefix $\tree_1$ of $\tree(\dtmc)$ cannot be extended to be
    a PT in $P$, contradiction. 
  \end{enumerate}
\end{enumerate}
\vspace*{-0.5cm}
\end{proof}

\begin{replemma}{lem:closure PCTL}
The closure formula of a $\pctlf$-formula equals:
$$
\begin{array}{rcl}
\closure(\Phi^a) & = & \Phi^a \\
\closure(\P{\bowtie q}{\X\Phi^a}) & = & {\P{\bowtie q}{\X\Phi^a}} \mbox{ for } \bowtie \, \in \{\leq, \geq\} \\
\closure(\P{\le q}{\Phi^a_1\U\Phi^a_2}) & = & {\P{\le q}{\Phi^a_1\U\Phi^a_2}} \\
\closure(\P{\ge q}{\Phi^a_1\U\Phi^a_2}) & = & {\P{\geq q}{\Phi^a_1\W\Phi^a_2}} \\
\closure(\P{\ge q}{\Phi^a_1\W\Phi^a_2}) & = & {\P{\ge q}{\Phi^a_1\W\Phi^a_2}} \\
\closure(\P{\le q}{\Phi^a_1\W\Phi^a_2}) & = & {\P{\le q}{\Phi^a_1\U\Phi^a_2}} \\
\closure({\Phi_1\lor\Phi_2}) & = & \closure({\Phi_1})\lor\closure({\Phi_2}).
\end{array}
$$
\end{replemma}
\begin{proof} \
  \begin{enumerate}
  \item $\closure(\Phi^a)=\Phi^a$.\\ 
    This case is trivial, since $\Phi^a$ only concerns atomic propositions.
  \item $\closure(\P{\bowtie q}{\X\Phi^a}) =  {\P{\bowtie q}{\X\Phi^a}} \mbox{ for } \bowtie \, \in \{\leq, \geq\}.$\\
As the proofs of these two cases is similar, we only consider $\bowtie \, = \, \geq$. 
Since $\Phi\subseteq\closure(\Phi)$ by Lemma~\ref{lem:topological closure}, it suffices to show $\closure(\Phi)\subseteq \Phi$.
Let $\Phi=\P{\ge q}{\X\Phi^a}$.
For any $\tree\in\closure(\Phi)$, the probability of reaching $\Phi^a$-states in one step is $\ge q$, therefore $\tree\in\Phi$.
\item $\closure(\P{\ge q}{\Phi^a_1\U\Phi^a_2}) = {\P{\ge q}{\Phi^a_1\W\Phi^a_2}}$.\\
Let $P_\U=\P{\ge q}{\Phi^a_1\U\Phi^a_2}$ and $P_\W=\P{\ge q}{\Phi^a_1\W\Phi^a_2}$.
 We first show $\closure(P_\U)\subseteq P_\W$. Let $\tree\in\closure(P_\U)$.
 Then for any $\tree_1\in\finitePrefix(\tree)$, there exists $\tree_2\in P_\U$
 such that $\tree_1\prefixTree\tree_2$. This means in each $\tree_1$, the probability of 
 reaching $(\Phi^a_1\lor\Phi^a_2)$-states via $\Phi^a_1$-states is $\ge q$.
 This indicates that $\tree\in P_\W$. Since otherwise there exists $\tree_1\in\finitePrefix(\tree)$
 such that in $\tree_1$ the probability of 
 reaching $\neg(\Phi^a_1\lor\Phi^a_2)$-states is $>1-q$, i.e.,
 $\tree_2\not\in P_\U$ for any $\tree_2\in\alltrees^\omega$ with $\tree_1\prefixTree\tree_2$.


Secondly we show that $P_\W\subseteq\closure(P_\U)$. Let $\tree\in P_\W$.
Then in any $\tree_1\in\finitePrefix(\tree)$, the probability of reaching $(\Phi^a_1\lor\Phi^a_2)$-states
via $\Phi^a_1$-states is $\ge q$. We can extend all nodes in $\tree_1$ to a node satisfying $\Phi^a_2$ with
probability 1. Thus the resulting PT is for sure in $P_\U$. According to Definition~\ref{def:topological closure},
$\tree\in\closure(P_\U)$.
\item $\closure(\P{\le q}{\Phi^a_1\U\Phi^a_2})  =  {\P{\le q}{\Phi^a_1\U\Phi^a_2}}$.\\
Case 3) indicates that properties like $\P{\ge q}{\Phi^a_1\W\Phi^a_2}$ are safety properties, 
hence properties of the form $\P{\le q}{\Phi^a_1\U\Phi^a_2}$ are also safety properties due to duality. 
Therefore $\closure(\P{\le q}{\Phi^a_1\U\Phi^a_2})=\P{\le q}{\Phi^a_1\U\Phi^a_2}$.
\item $\closure(\P{\ge q}{\Phi^a_1\W\Phi^a_2})  = {\P{\ge q}{\Phi^a_1\W\Phi^a_2}}$, and\\
$\closure(\P{\le q}{\Phi^a_1\W\Phi^a_2}) = {\P{\le q}{\Phi^a_1\U\Phi^a_2}}$. \\
The proofs of these cases are similar to case 3) and 4).
 \item $\closure({\Phi_1\lor\Phi_2}) = \closure({\Phi_1})\lor\closure({\Phi_2})$.\\
Straightforward from Def.~\ref{def:topological closure} and Lemma~\ref{lem:topological closure}.
  \end{enumerate}
  \vspace*{-0.5cm}
\end{proof}

\begin{reptheorem}{thm:correctness of algorithm}
  Algorithm~\ref{alg:safety and liveness} is correct.  
\end{reptheorem}
\begin{proof}
  Line~\ref{line:distributivity} is justified by the distribution rules of conjunction and disjunction. 
The correctness of Line~\ref{line:safety} and Line~\ref{line:liveness} is guaranteed by Lemma~\ref{lem:closure PCTL} 
and Proposition~\ref{prop:decompose} respectively, 
while the correctness of Line~\ref{line:conjunction} is ensured by Lemma~\ref{lem:disjunction and conjunction}.
\end{proof}

\begin{reptheorem}{thm:safety pctl}
Every $\pctls$-formula is a safety property.
\end{reptheorem}
\begin{proof}
Let $\Phi$ be a $\pctls$-formula.
It suffices to show that $\closure(\Phi)\subseteq\Phi$.
The proof is by structural induction on $\Phi$. 
  \begin{enumerate}
  \item $\Phi = \Phi^a$. This case is trivial.
  \item $\Phi = \Phi_1\land\Phi_2$ or $\Phi = \Phi_1\lor\Phi_2$. Since $\Phi_1$ and $\Phi_2$ are 
    safety properties by induction hypothesis, 
    $\Phi$ is a safety property, i.e., $\Phi = \closure(\Phi)$ by Lemma~\ref{lem:disjunction and conjunction}.
  \item $\Phi = \P{\ge q}{\X\Phi'}$, where $\Phi'$ is a safety property by induction hypothesis.
    If $\Phi$ is not a safety property, there exists 
    $\tree\not\in\Phi$, but for all $\tree_1\in\finitePrefix(\tree)$, there exists 
    $\tree_2\in\Phi$ such that $\tree_1\prefixTree\tree_2$.
    This indicates that there exists $\tree'\not\in\Phi'$ (by omitting the first node of $\tree$), but for any $\tree'_1\in\finitePrefix(\tree')$,
    there exists $\tree'_2\in\Phi'$ such that $\tree'_1\prefixTree\tree'_2$,
    which contradicts that $\Phi'$ is a safety property.
\item $\Phi = \P{\ge q}{\Phi_1\W\Phi_2}$, where $\Phi_1,\Phi_2\in\pctls$. 
    By contraposition. Assume there exists $\tree\in\closure(\Phi)$ such that
    $\tree\not\in\Phi$. Therefore $\tree\in\P{<q}{\Phi_1\W\Phi_2}$, i.e., $\tree\in\P{>1-q}{(\Phi_1\land\neg\Phi_2)\U(\neg\Phi_1\land\neg\Phi_2)}$
    due to duality. Since $\Phi_1$ and $\Phi_2$ are safety properties by induction hypothesis, 
    so is $\Phi_1\lor\Phi_2$ by Lemma~\ref{lem:disjunction and conjunction}.
    Therefore for PTs not in $\Phi_1\lor\Phi_2$, finite counterexamples for $\Phi_1\lor\Phi_2$ (or witness for $\neg\Phi_1\land\neg\Phi_2$) 
    always exist by Theorem~\ref{thm:counterexample safety and liveness}.
    Thus there exists $\tree_1\in\finitePrefix(\tree)$ of which the probability satisfying $(\Phi_1\land\neg\Phi_2)\U(\neg\Phi_1\land\neg\Phi_2)$
    exceeds $1{-}q$. As a result $\tree_2\not\in\Phi$ for any $\tree_2\in\alltrees^\omega$ such that $\tree_1\prefixTree\tree_2$, which implies
    $\tree\not\in\closure(\Phi)$. Contradiction.
\item $\Phi = \P{\leq q}{\Phi_1\U\Phi_2}$, where $\neg\Phi_1,\neg\Phi_2\in\pctls$.
    By induction hypothesis, $\neg\Phi_1$ and $\neg\Phi_2$ are safety properties.
    By contraposition. Assume $\Phi$ is not a safety property. Then there exists
    $\tree\not\in\Phi$, i.e., $\tree\in\P{>q}{\Phi_1\U\Phi_2}$ 
    such that for each $\tree_1\in\finitePrefix(\tree)$,
    there exists $\tree_2\in\Phi$ with $\tree_1\prefixTree\tree_2$.  
    Since $\neg\Phi_2$ is a safety property, for each $\tree'\not\in\neg\Phi_2$, i.e., $\tree'\in\Phi_2$,
    there exists $\tree'_1\in\finitePrefix(\tree')$ such that $\tree'_2\not\in\neg\Phi_2$, i.e., $\tree'_2\in\Phi_2$
    for all $\tree'_2\in\alltrees^\omega$ with $\tree'_1\prefixTree\tree'_2$. In other words,
    for each $\tree'\in\Phi_2$, there always exists a finite-depth witness (or finite counterexamples for $\neg\Phi_2$) such that 
    it is enough to check this witness in order to guarantee $\tree'\in\Phi_2$, similarly for $\Phi_1$.
    Therefore in case $\tree\in\P{>q}{\Phi_1\U\Phi_2}$, there exists $\tree_1\in\finitePrefix(\tree)$,
    where the probability satisfying $\Phi_1\U\Phi_2$ exceeds $q$. This means 
    $\tree_2\not\in\Phi$ for all $\tree_2\in\alltrees^\omega$ if $\tree_1\prefixTree\tree_2$. Contradiction.
  \end{enumerate}
  \vspace*{-0.5cm}
\end{proof}

\begin{reptheorem}{thm:safety pctl complete}
For every safety property $\Phi$ expressible in \emph{PCTL} (no strict bounds),
there exists $\Phi'\in\pctls$ with $\Phi\equiv\Phi'$.
\end{reptheorem}
\begin{proof}
The goal is to prove for any PCTL formula $\Phi$, either $\Phi$ is not a safety property,
or there exists a formula in $\pctls$ equivalent to $\Phi$. The proof is by structural
induction on $\Phi$. Cases for atomic propositions and Boolean connections are simple
and omitted. 
\begin{enumerate}
\item 
Let $\Phi=\P{\ge q}{\X\Phi'}$ with $q>0$ (in case $q=0$, $\Phi\equiv\top\in\pctls$). 
Suppose $\Phi$ is a safety property, but $\Phi'$ is not.
According to Def.~\ref{def:safety}, there exists $\tree'\not\in\Phi'$,
but for all $\tree'_1\in\finitePrefix(\tree')$, there exists $\tree'_2\in\Phi'$
such that $\tree'_1\prefixTree\tree'_2$.
As a result, there also exists $\tree\in\P{\le 0}{\X\Phi'}$,
but for all $\tree_1\in\finitePrefix(\tree)$, there exists $\tree_2\in\P{\ge 1}{\X\Phi'}$
such that $\tree_1\prefixTree\tree_2$.
Since $\tree\in\P{\le 0}{\X\Phi'}$ implies $\tree\not\in\Phi$, and $\tree\in\P{\ge 1}{\X\Phi'}$
implies $\tree\in\Phi$, we conclude that $\Phi$ is not a safety property. Contradiction.
\item Let $\Phi=\P{\le q}{\Phi_1\U\Phi_2}$. Suppose $\Phi$ is a safety property, we prove that both $\neg\Phi_1$ and $\neg\Phi_2$ must be 
safety properties. By contraposition. Suppose at least one of $\neg\Phi_1$ and $\neg\Phi_2$ is not a safety property.
Since $\Phi$ is a safety property. By Def.~\ref{def:safety}, 
for each $\tree\not\in\Phi$, i.e., $\tree\in\P{> q}{\Phi_1\U\Phi_2}$,
there exists $\tree_1\in\finitePrefix(\tree)$ such that $\tree_2\not\in\Phi$ for each $\tree_2\in\alltrees^\omega$ with
$\tree_1\prefixTree\tree_2$. In other words, in the finite-depth prefix $\tree_1$ of $\tree$, the probability
of paths satisfying $\Phi_1\U\Phi_2$ already exceeds $q$. As such finite-depth prefix always exists for each $\tree\not\in\Phi$, this indicates that
properties $\Phi_1$ and $\Phi_2$ always have finite-depth witnesses. Equivalently, properties $\neg\Phi_1$ and $\neg\Phi_2$ always
have finite counterexamples. By Theorem~\ref{thm:counterexample safety and liveness}, 
$\neg\Phi_1$ and $\neg\Phi_2$ are safety properties. Contradiction.
\item Other cases are similar. For instance let $\Phi=\P{\ge q}{\Phi_1\W\Phi_2}$, where either $\Phi_1$ or $\Phi_2$ 
is a not safety property. By duality, $\Phi\equiv\P{\ge 1-q}{(\Phi_1\land\neg\Phi_2)\U(\neg\Phi_1\land\neg\Phi_2)}$.
By induction, $\neg\Phi_1\land\neg\Phi_2$ is not a safety property. The remaining proof is the same as the case above.
\end{enumerate}
\vspace*{-0.5cm}
\end{proof}

\begin{reptheorem}{thm:liveness pctl sound}
  Every $\pctll^{<}$-formula is a liveness property.
\end{reptheorem}
\begin{proof}
  We prove a stronger result by considering also properties with
  strict probabilistic bounds.
  Let $\Phi \in \pctll^{<}$.
  It suffices to prove $\closure(\Phi)=\alltrees^\omega$.
  This is done by structural induction on $\Phi$. 
  \begin{enumerate}
  \item $\Phi = \top$. Trivial.
  \item $\Phi = \P{\unrhd q}{\Diamond\Phi^a}$. Let $\tree\in\alltrees^\omega$.
    For any $\tree_1\in\finitePrefix(\tree)$, we can extend it to a tree $\tree_2$
    by letting $\pi\concat(1,A)\in\tree_2$ for each $\pi\in\tree_1$, where $A\subseteq\AP$ and $A\models\Phi^a$.
    After doing so, the probability of $\tree_2$ satisfying $\Diamond\Phi^a$ is 1, so $\tree_2\in\Phi$.
    Therefore $\tree\in\closure(\Phi)$.
  \item  These cases for conjunction and disjunction can be proved by applying induction hypothesis and Lemma~\ref{lem:disjunction and conjunction}.
    Now let $\Phi = \P{\unrhd q}{\Phi_1\W\Phi_2}$, where either $\Phi_1\in\pctll^{<}$ or $\Phi_2\in\pctll^{<}$.
    In case $\Phi_2$ is a liveness property, for any $\tree_1\in\alltrees^*$, there exists $\tree_2\in\Phi_2$ such that $\tree_1\prefixTree\tree_2$.
    Since $\tree_2\in\Phi_2$ implies $\tree_2\in\Phi$, $\Phi$ is live.
    Assume $\Phi_1\in\pctll^{<}$ and $\Phi_2\not\in\pctll^{<}$. Let $\tree_1\in\alltrees^*$ be an arbitrary finite-depth tree.
    Since $\Phi_1$ is a liveness property by induction hypothesis, $\Phi_1\neq\bot$.
    Let $\tree_2\in\alltrees^\omega$ with $\tree_1\prefixTree\tree_2$ such that all leaves of $\tree_1$
    are appended with a PT in $\Phi_1$. Since $\Phi_1$ is a liveness property, 
    all nodes in $\tree_2$ satisfying $\Phi_1$, namely, $\tree_2\in\P{\ge 1}{\Phi_1\W\bot}$. This implies
    $\tree_2\in\Phi$ and $\Phi$ is a liveness property.
  \item $\Phi = \P{\unrhd q}{\X\Phi'}\in\pctll^{<}$, where $\Phi'\in\pctll^{<}$.
    By induction hypothesis, $\Phi'$ is a liveness property. Thus
    for any $\tree'_1\in\alltrees^*$, there exists $\tree'_2\in\Phi'$ such that
    $\tree'_1\prefixTree\tree'_2$. This implies for any $\tree_1\in\alltrees^*$,
    there exists $\tree_2\in\alltrees^\omega$ such that $\tree_1\prefixTree\tree_2$,
    and the probability of satisfying $\X\Phi'$ is equal to 1. Thus $\tree_2\in\Phi$,
    and $\Phi$ is a liveness property.
  \item $\Phi = \P{\unrhd q}{\Phi_1\U\Phi_2}$ where $\Phi_2\in\pctll^{<}$.
    By induction hypothesis, $\Phi_2$ is a liveness property.
    Thus for any $\tree_1\in\alltrees^*$,
    there exists $\tree_2\in\Phi_2$ and $\tree_1\prefixTree\tree_2$.
    Note $\tree_2\in\Phi_2$ implies $\tree_2\in\Phi$. Therefore $\Phi$ is a liveness property.
 \end{enumerate}
  \vspace*{-0.5cm}
\end{proof}

\begin{reptheorem}{thm:liveness pctl complete}
  For every liveness property $\Phi$ expressible in \emph{PCTL}, there exists $\Phi'\in\pctll^{>}$
  with $\Phi\equiv\Phi'$.
\end{reptheorem}
\begin{proof}
  As in the proof of Theorem~\ref{thm:liveness pctl sound}, we also
  consider properties with strict probabilistic bounds.
  Let $\Phi$ be an arbitrary PCTL property. We prove that either
  $\Phi$ is not a liveness property, or there is $\Phi'\in\pctll^{>}$ such that $\Phi\equiv\Phi'$.
  The proof is by structural induction
  on $\Phi$. Here we only show the proof of a few cases, while all other cases are similar.
  \begin{enumerate}
  \item $\Phi=\P{\unrhd q}{\Phi_1\U\Phi_2}$.
    In case $\Phi_1$ is a liveness property. By induction hypothesis, there exists $\Phi'_1\in\pctll^{>}$
    such that $\Phi_1\equiv\Phi'_1$. Thus $\Phi\equiv\P{\unrhd q}{\Phi'_1\U\Phi_2}\in\pctll^{>}$.
    The case when $\Phi_2$ is a liveness property is similar.
    Now we assume that neither $\Phi_1$ nor $\Phi_2$ is live. Therefore
    $\Phi_1\lor\Phi_2$ is not a liveness property. By Def.~\ref{def:liveness},
    there exists $\tree_1\in\alltrees^*$ such that $\tree_2\not\in\Phi_1\lor\Phi_2$,
    for any $\tree_2\in\alltrees^\omega$ with $\tree_1\prefixTree\tree_2$. Since
    $\tree_2\not\in\Phi_1\lor\Phi_2$ implies $\tree_2\not\in\Phi$, which contradicts 
    with the assumption that $\Phi$ is a liveness property. 
  \item $\Phi=\P{\unrhd q}{\Phi_1\W\Phi_2}$.
    This case can be proved in a similar way as the case above.
  \item $\Phi=\P{\unrhd q}{\X\Phi_1}$.
    In case $\Phi_1$ is a liveness property, there exists $\Phi'_1\in\pctll^{>}$ such that
    $\Phi_1\equiv\Phi'_1$ by induction hypothesis. Thus $\Phi\equiv\P{\unrhd q}{\X\Phi'_1}\in\pctll^{>}$.
    Assume $\Phi_1$ is not live.
    According to Def.~\ref{def:liveness}, there exists $\tree_1\in\alltrees^*$
    such that $\tree_2\not\in\Phi_1$ for any $\tree_2\in\alltrees^\omega$
    with $\tree_1\prefixTree\tree_2$. Let $\tree'_1$ be the PT such that after one step it will
    perform like $\tree_1$ with probability one. Then for any $\tree'_2\in\alltrees^\omega$, we have
    $\tree'_2\in\P{\le 0}{\X\Phi_1}$, provided $\tree'_1\prefixTree\tree'_2$. 
    This implies $\tree'_2\not\in\Phi$ and $\Phi$ is not a liveness property. 
  \end{enumerate}
  \vspace*{-0.5cm}
\end{proof}

\begin{reptheorem}{thm:simulation logical}
 $\precsim~=~\epctls~=~\epctlltwo~\subsetneq~\epctllone$.
\end{reptheorem}
\begin{proof}
  Since we consider MCs without absorbing states (i.e., states without any outgoing transitions), it follows
  from Prop.~18 and Theorem~48 in~\cite{Baier2005CBS} that $\precsim\, \subseteq \, \equiv_{\mathit{PCTL}}$,
  where $s_1 \, \equiv_{\mathit{PCTL}} \, s_2$ iff $s_1\models\Phi$ implies $s_2\models\Phi$ and vice versa,
  for any PCTL property $\Phi$. 
  Therefore $\precsim~\subseteq~\epctls$, $\precsim~\subseteq~\epctllone$, 
  and $\precsim~\subseteq~\epctlltwo$. It suffices to prove the following cases:
  \begin{enumerate}
  \item $\epctls \, \subseteq \, \precsim$. 
    Let $\pctls^{2005}$ denote the PCTL safety fragment of~\cite{Baier2005CBS}, for which
    it is known that $\precsim_{\pctls^{2005}}\, \subseteq \, \precsim$. Since $\pctls^{2005}$ is
    a subset of $\pctls$, $\epctls\, \subseteq \, \precsim$.
  \item $\precsim~\subsetneq~\epctllone$. It suffices to show that there exists $s_1$ and $s_2$ such
    that $s_1~\not\precsim~s_2$, but $s_1~\epctll~s_2$. Let $\AP=\{a\}$ and $s_1$ and $s_2$ be two states
    such that $L(s_1)=\{a\}$ and $L(s_2)=\emptyset$. Moreover $s_1\rightarrow\dirac{s_2}$ and
    $s_2\rightarrow\dirac{s_1}$, where $\dirac{s_i}$ denotes Dirac distributions, i.e., $\dirac{s_i}(s_i)=1$. 
    By Def.~\ref{def:simulation}, $s_1~\not\precsim~s_2$ since $L(s_1)\neq L(s_2)$.
    Now we show that for each $\Phi\in\pctll^{<}$: $s_1,s_2\models\Phi$. We prove by structural induction on $\Phi$.
    \begin{enumerate}
    \item $\Phi\equiv\top$. Trivial.
    \item $\Phi\equiv\P{\ge q}{\Diamond\Phi^a}$. Since $\AP=\{a\}$, either $\Phi^a=a$ or $\Phi^a=\neg a$. In both cases,
      we have $s_1,s_2\models\Phi$.
    \item $\Phi\equiv\Phi_1\land\Phi_2$. By the definition of
      $\pctll^{<}$, $\Phi_1,\Phi_2\in\pctll^{<}$, therefore
      $s_1,s_2\models\Phi_1$ and $s_1,s_2\models\Phi_2$ by induction
      hypothesis, which implies $s_1,s_2\models\Phi$.
    \item $\Phi\equiv\Phi_1\lor\Phi_2$. By the definition of
      $\pctll^{<}$, at least one of $\Phi_1$ and $\Phi_2$ is in
      $\pctll^{<}$. Suppose $\Phi_1\in\pctll^{<}$. Then
      $s_1,s_2\models\Phi_1$ by induction
      hypothesis, which implies $s_1,s_2\models\Phi$. The case for
      $\Phi_2\in\pctll^{<}$ is similar and omitted here.
    \item $\Phi\equiv\P{\ge q}{\X\Phi'}$. By the definition of
      $\pctll^{<}$, $\Phi'\in\pctll^{<}$. Thus $s_1,s_2\models\Phi'$ by
      induction hypothesis, which implies $s_1,s_2\models\P{\ge
        q}{\X\Phi'}$. Therefore $s_1,s_2\models\Phi$.
    \item $\Phi\equiv\P{\ge q}{\Phi_1\U\Phi_2}$. By the definition of
      $\pctll^{<}$, $\Phi_2\in\pctll^{<}$. Thus $s_1,s_2\models\Phi_2$ by
      induction hypothesis, which implies $s_1,s_2\models\Phi$.
    \item $\Phi\equiv\P{\ge q}{\Phi_1\W\Phi_2}$. By the definition of
      $\pctll^{<}$, at least one of $\Phi_1$ and $\Phi_2$ is in
      $\pctll^{<}$. Suppose $\Phi_1\in\pctll^{<}$. Then $s_1,s_2\models\Phi_1$
      by induction hypothesis.  Therefore from $s_1$ and $s_2$,
      property $\Phi_1$ will always be satisfied with probability 1.
      Thus $s_1,s_2\models\P{\ge 1}{\Phi_1\W\Phi_2}$, which implies
      $s_1,s_2\models\Phi$. This case for $\Phi_2\in\pctll^{<}$ can be
      proved similarly and is omitted.
    \end{enumerate}
  \item $\epctlltwo~\subseteq~\precsim$. Let $\dtmc$ be an MC and
    $s_1~\epctlltwo~s_2$.  We show that $s_1~\precsim~s_2$.  Let
    $\Phi_1\in\pctll^{>}$ such that $s_2\not\models\Phi_1$.  We argue that such
    $\Phi_1$ always exists. 
    \begin{itemize}
    \item Let $\Phi^a$ be a literal formula such that $s_2\not\models\P{\ge 1}{\Box\P{\ge 1}{\Diamond\Phi^a}}$.
      Since $\P{\ge 1}{\Box\P{\ge 1}{\Diamond\Phi^a}}\in\pctll^{>}$, we can simply let $\Phi_1=\P{\ge 1}{\Box\P{\ge 1}{\Diamond\Phi^a}}$.
    \item 
     If such $\Phi^a$ does not exist, i.e.,  $s_2\models\P{\ge 1}{\Box\P{\ge 1}{\Diamond\Phi^a}}$
    for any literal formula $\Phi^a$. Then it must be the case that $s_2$ belongs to a bottom
    strongly connected component $\mathit{BSCC}$ (the maximal set of
    states which are reachable from each other and have no
    transitions going to states not in the $\mathit{BSCC}$) such that for each
    $\Phi^a$ there exists a state $s$ in $\mathit{BSCC}$ with
    $s\models\Phi^a$. Therefore there must exist $\Phi^b$ such that $s\not\models\P{\ge
      1}{\Box\Phi^b}$ for all states $s$ in $\mathit{BSCC}$.
    Let $\Phi_1=\P{\ge 1}{\Phi'_1\U\P{\ge
        1}{\Box\Phi^b}}$ for any $\Phi'_1\in\pctll^{>}$ such that
    $\Phi_1\not\equiv\bot$. It follows that $\Phi_1\in\pctll^{>}$ and $s_2\models\Phi_1$.
    \end{itemize}
As a result, $\Phi_1\in\pctll^{>}$ and
    $s_2\not\models\Phi_1$.  Let $\Phi_2$ be an arbitrary PCTL
    property. We have $s_2\models\Phi_1\lor\Phi_2$ iff
   $s_2\models\Phi_2$. 
    Moreover $\Phi_1\in\pctll^{>}$, so $\Phi_1\lor\Phi_2\in\pctll^{>}$.
   If $s_1~\not\precsim~s_2$, there exists $\Phi_2$ such that
    $s_1\models\Phi_2$ but $s_2\not\models\Phi_2$
    by~\cite{Baier2005CBS}.  Therefore $s_1\models\Phi_1\lor\Phi_2$
    and $s_2\not\models\Phi_1\lor\Phi_2$, which contradicts the fact
    that $s_1~\epctlltwo~s_2$.
  \end{enumerate}
  \vspace*{-0.5cm}
\end{proof}

\begin{reptheorem}{thm:strong sound and complete}
Every $\pctlss$-formula is a strong safety property, and for any strong safety property $\Phi$ expressible in \emph{PCTL},
there exists $\Phi'\in\pctlss$ with $\Phi\equiv\Phi'$.
\end{reptheorem}
\begin{proof}
First, for any $\Phi\in\pctlss$, we prove that $\Phi$ satisfies the three conditions
 in Def.~\ref{def:strong safety}.
 \begin{enumerate}
 \item $\Phi$ is a safety property. Since $\pctlss~\subset\pctls$,
   $\Phi$ is a safety property by Theorem~\ref{thm:safety pctl}.
 \item $\Phi$ is closed under stuttering and shrinking. We prove by structural induction on $\Phi$.
   \begin{enumerate}
   \item $\Phi=\Phi^a$. Trivial.
   \item $\Phi=\Phi_1\land\Phi_2$ or $\Phi=\Phi_1\lor\Phi_2$ with $\Phi_1,\Phi_2\in\pctlss$. 
     By induction hypothesis,
     $\Phi_1$ and $\Phi_2$ are closed under stuttering and shrinking. Therefore $\Phi$ is also closed
     under stuttering and shrinking.
   \item $\Phi=\P{\ge q}{\Phi_1\W\Phi_2}$, where $\Phi_1\in\pctlss$ and $\Phi_2\in\mathcal{F}^\Box$ as defined in Def.~\ref{def:strong safety}.
     By induction hypothesis, both $\Phi_1$ and $\Phi_2$ are strong safety properties and closed under stuttering and shrinking.
     For any $\tree\in\Phi$, it is easy to see that all PTs obtained by stuttering or shrinking $\tree$ for finite steps are
     also in $\Phi$. The only non-trivial case is when we delete the first nodes in $\tree$ satisfying $\Phi_2$.
     Since $\Phi_2\in\mathcal{F}^\Box$, $\tree'\in\Phi_2$ implies $\tree'\in\P{\ge 1}{\X\Phi_2}$ for all $\tree'\in\alltrees^\omega$.
     Therefore all suffixes of $\tree'$ are also in $\Phi_2$, as long as $\tree'\in\Phi_2$.
     Even after deleting the first nodes satisfying $\Phi_2$ in $\tree$, we still have $\tree\in\Phi$.
     Thus $\Phi$ is a strong safety property.
   \end{enumerate}
 \end{enumerate}

Secondly, let $\Phi$ be a safety property in PCTL. We show that either $\Phi$ is not a strong safety property, or there exists $\Phi'\in\pctlss$ 
  such that $\Phi\equiv\Phi'$. 
  We proceed by structural induction on $\Phi$.
  \begin{enumerate}
  \item $\Phi=\Phi^a$. Trivial.
  \item $\Phi=\Phi_1\lor\Phi_2$. 
    Let $\stsh(\Phi)$ be the smallest set containing all PTs in $\Phi$
    and closed under stuttering and shrinking.  By
    Def.~\ref{def:strong safety}, $\Phi=\stsh(\Phi)$ iff $\Phi$ is a
    strong safety property.  Assume $\Phi$ is a strong safety
    property (otherwise trivial). Then $\Phi=\stsh(\Phi)$.  Since
    $\Phi_1\lor\Phi_2\subseteq\stsh(\Phi_1)\cup\stsh(\Phi_2)\subseteq\stsh(\Phi_1\lor\Phi_2)=\Phi_1\lor\Phi_2$,
    $\Phi_1\lor\Phi_2=\stsh(\Phi_1)\cup\stsh(\Phi_2)$. 
    Since $\stsh(\Phi_1)$ and $\stsh(\Phi_2)$ are strong safety properties,
    by induction hypothesis
    there exists $\Phi'_1,\Phi'_2\in\pctlss$ such that
    $\Phi'_1\equiv\stsh(\Phi_1)$ and $\Phi'_2\equiv\stsh(\Phi_2)$. In
    other words, $\Phi\equiv\Phi'=\Phi'_1\lor\Phi'_2\in\pctlss$, whenever
    $\Phi$ is a strong safety property. 
    The case $\Phi=\Phi_1\land\Phi_2$ can be proven in a similar way.
 \item $\Phi=\P{\ge q}{\X\Phi_1}$. 
   Let $q>0$ (otherwise $\Phi\equiv\top\in\pctlss$). 
According to Def.~\ref{def:strong safety}, $\Phi\not\in\pctlss$.
    Assume $\Phi$ is a strong safety property.
    Let $\tree\in\alltrees^\omega$ such that $\tree\in\Phi$ but $\tree\not\in\Phi_1$. 
    \begin{itemize}
   \item 
Firstly, assume such $\tree$ exists. Then by repeating the first node of $\tree$, 
    the probability of satisfying $\Phi_1$ in the next step is 0, which means that $\Phi$ is not closed
    under stuttering, thus is not a strong safety property.
\item    Secondly, suppose such $\tree$ does not exist, i.e., $\Phi$ implies $\Phi_1$.
    For any finite-depth PT $\tree_1\in\alltrees^*$,
    we append each leaf of $\tree_1$ with a PT in $\Phi_1$.
    After doing so, each node in $\tree_1$ will go to nodes satisfying $\Phi_1$ with probability one in one step,
    i.e., the resulting PT satisfies $\Phi$, which implies that $\Phi_1$ is satisfied.
    By Def.~\ref{def:liveness}, $\Phi_1$ is a liveness property. Since $\Phi_1$ is also a safety property,
    it is only possible when $\Phi_1\equiv\top$, which implies $\Phi\equiv\top\in\pctlss$.
   \end{itemize}
   \item $\Phi=\P{\ge q}{\Phi_1\W\Phi_2}$.\label{item:ss complete W}
    We distinguish two cases:
    \begin{itemize}
    \item Suppose either $\Phi_1$ or $\Phi_2$ is not a strong safety property.
     Hence $\Phi_1\lor\Phi_2$ is not a strong safety property either.
      Then there exists $\tree\in\Phi$ which implies $\tree\in\Phi_1\lor\Phi_2$, but 
      there is $\tree'\not\in(\Phi_1\lor\Phi_2)$ obtained by stuttering or shrinking $\tree$
      for finite steps. Since $\tree'\not\in(\Phi_1\lor\Phi_2)$ implies $\tree'\not\in\Phi$,
      $\Phi$ is not a strong safety property.
    \item Suppose $\Phi_1,\Phi_2$ are strongly safe. Suppose there exists
      no $\Phi'_2\in\mathcal{F}^\Box$ such that $\Phi\equiv\P{\ge
        q}{\Phi_1\W\Phi'_2}$ (otherwise trivial). Since $\Phi_2\not\in\mathcal{F}^\Box$,
      there exists $\tree_2\in\Phi_2$, but $\tree_2\not\in\P{\ge
        1}{\X\Phi_2}$.  Otherwise $\Phi_2\equiv\P{\ge
        1}{\Box\Phi_2}\in\mathcal{F}^\Box$.  Let
      $\tree\in\P{=q}{\Phi_1\W\Phi_2}$ such that the probability of
      $\tree$ reaching some suffixes $\tree'\in\Phi_2$ is exactly
      equal to $q$, clearly $\tree\in\Phi$.  Let $\tree_2\in\Phi_2$
      and $\tree_2\not\in\P{\ge 1}{\X\Phi_2}$, i.e.,
      $\tree_2\in\P{<1}{\X\Phi_2}$.  The maximal probability of
      $\tree_2$ satisfying $\Phi_2$ in the next step is equal to
      $q'<1$.  By removing the initial node of $\tree_2$ from $\tree$
      (this is allowed, since the initial node of $\tree_2$ is not the
      initial node of $\tree$), the probability of satisfying
      $\Phi_1\W\Phi_2$ in the resulting PT $\tree'$ is equal to
      $q\times q'< q$. Therefore $\tree'\not\in\Phi$, and $\Phi$ is
      not a strong safety property.
    \end{itemize}
  \item $\Phi=\P{\le q}{\Phi_1\U\Phi_2}$. Using duality laws,
    $\Phi\equiv\P{\ge
      1-q}{(\Phi_1\land\neg\Phi_2)\W(\neg\Phi_1\lor\neg\Phi_2)}$,
    which is a strong safety property iff there exists
    $\Phi'_1\in\pctlss$ and $\Phi'_2\in\mathcal{F}^\Box$ such that
    $\Phi'_1\equiv(\Phi_1\land\neg\Phi_2)$ and
    $\Phi'_2\equiv(\neg\Phi_1\land\neg\Phi_2)$ according to the proof
    of case~\ref{item:ss complete W}).
  \end{enumerate}
  \vspace*{-0.5cm}
\end{proof}

\begin{replemma}{lem:absolute is liveness}
Every absolute liveness property is live.
\end{replemma}
\begin{proof}
  By contraposition. Let $P$ be an absolute liveness property, but
  $P$ is not a liveness property. By Def.~\ref{def:liveness},
  there exists $\tree_1\in\alltrees^*$ such that $\tree_2\not\in P$ for all $\tree_2\in\alltrees^\omega$
  with $\tree_1\prefixTree\tree_2$. 
  Let $\tree'\in\alltrees^\omega$ such that $\tree_1\in\finitePrefix(\tree')$ and $\tree$ being a suffix of $\tree'$ for some $\tree\in P$.
  By construction, $\tree'\not\in P$, which contradicts that $P$
  is an absolute liveness property.
\end{proof}

\begin{replemma}{lem:stable and absolute}
  For any $P\neq\alltrees^\omega$, $P$ is a stable property iff $\comp{P}$ is
  an absolute liveness property.
\end{replemma}
\begin{proof}
  The proof is similar to the proof of~\cite{Sistla1994SLF}[Lemma 2.1], which is rephrased here for completeness.
  We prove directly by using Def.~\ref{def:stable} and \ref{def:absolute liveness}.
  First, let $P\neq\alltrees^\omega$ be a stable property. We prove that $\comp{P}$ is
  an absolute liveness property. By contraposition. Assume $\comp{P}\neq\emptyset$ is not an absolute
  liveness property, i.e., there is $\tree\in\comp{P}$ such that $\tree'\not\in\comp{P}$
  with $\tree$ is a suffix of $\tree'$. In other words, $\tree'\in P$ and 
  $\tree\not\in P$, where $\tree$ is a suffix of $\tree'$, which indicates
  that $P$ is not stable.

  Secondly, let $P$ be an absolute liveness property. We show that $\comp{P}$
  is a stable property. By contraposition. Assume $\comp{P}$ is not a
  stable property.
  Thus there is $\tree\in\comp{P}$ such that $\tree'\not\in\comp{P}$ with $\tree'$
  a suffix of $\tree$. In other words, there exists $\tree'\in P$
  such that $\tree\not\in P$. Since $\tree'$ is a suffix of $\tree$, this contradicts
  the assumption that $P$ is an absolute liveness property.
\end{proof}

\begin{reptheorem}{thm:absolute sound and complete}
Every $\pctlal$-formula is an absolute liveness property, and for any
absolute liveness property $\Phi$ expressible in \emph{PCTL}, there exists
$\Phi'\in\pctlal$ with $\Phi\equiv\Phi'$. 
\end{reptheorem}
\begin{proof}
  First, let $\Phi\in\pctlal$. We prove by structural induction on $\Phi$ that
  $\Phi$ is an absolute liveness property.
  \begin{enumerate}
  \item $\Phi=\top$. Trivial.
  \item $\Phi=\Phi_1\land\Phi_2$ where $\Phi_1,\Phi_2\in\pctlal$.  Let
    $\tree\in\Phi$, which indicates $\tree\in\Phi_1$ and
    $\tree\in\Phi_2$.  By induction hypothesis, $\Phi_1$ and $\Phi_2$
    are absolute liveness properties.  Thus, for each
    $\tree'\in\alltrees^\omega$, we have $\tree'\in\Phi_1$ and
    $\tree'\in\Phi_2$, provided $\tree$ is a suffix of $\tree'$. Thus
    $\tree'\in\Phi_1\land\Phi_2$ as desired. The case
    $\Phi=\Phi_1\lor\Phi_2$ is similar.
  \item $\Phi=\P{>0}{\X\Phi'}$, where $\Phi'\in\pctlal$. By induction
    hypothesis, $\Phi'$ is an absolute liveness property. Let
    $\tree\in\Phi$. The probability of reaching trees $\tree'_1$ in
    one step is positive, where $\tree'_1\in\Phi'$ is a suffix of
    $\tree$.  Let $\tree'\in\alltrees^\omega$ such that $\tree$ is a
    suffix of $\tree'$.  Then in one step $\tree'$ will reach $\tree'_2$
    such that $\tree'_1$ is a suffix of $\tree'_2$.  Since $\Phi'$ is
    an absolute liveness property, $\tree'_2\in\Phi'$.  Therefore
    $\tree'\in\Phi$.
  \item $\Phi=\P{>0}{\Phi_1\U\Phi_2}$, where $\Phi_2\in\pctlal$ or
    $\Phi_1\in\pctlal$ and $\Phi_2\land\neg\Phi_1\equiv\bot$.  First
    assume $\Phi_2\in\pctlal$.  By induction hypothesis, $\Phi_2$ is
    an absolute liveness property.  Let $\tree\in\Phi$, then the
    probability of reaching some $\tree'_2\in\Phi_2$ is positive,
    where $\tree'_2$ is a suffix of $\tree$.  For any $\tree'$ such
    that $\tree$ is a suffix of $\tree'$, $\tree'_2$ is also a suffix
    of $\tree'$.  Since $\Phi_2$ is an absolute liveness property,
    $\tree'\in\Phi_2$, which implies $\tree'\in\Phi$.  Secondly,
    assume $\Phi_2\not\in\pctlal$, $\Phi_1\in\pctlal$, and
    $\Phi_2\land\neg\Phi_1\equiv\bot$.  By induction hypothesis,
    $\Phi_1$ is an absolute liveness property.  Let $\tree\in\Phi$. We
    have either $\tree\in\Phi_2$ or $\tree\in\Phi_1$.  Since
    $\Phi_2\land\neg\Phi_1\equiv\bot$, $\tree\in\Phi_2$ implies
    $\tree\in\Phi_1$.  We only need to consider the case when
    $\tree\in\Phi_1$.  Since $\Phi_1$ is an absolute liveness
    property, in case $\tree\in\Phi_1$, we have $\tree'\in\Phi_1$ for
    any $\tree'$, provided $\tree$ is a suffix of $\tree'$. With the
    assumption that $\tree\in\Phi$, we have $\tree'\in\Phi$.
  \item $\Phi=\P{>0}{\Phi_1\W\Phi_2}$, where $\Phi_1,\Phi_2\in\pctlal$, or $\Phi_1\in\pctlal$ and $\Phi_2\land\neg\Phi_1\equiv\bot$. 
   The proof for the case is similar as the above case and omitted here.
  \end{enumerate}

  Secondly, we prove that for any PCTL formula $\Phi$, either
  $\Phi$ is not an absolute liveness property, or 
  there exists $\Phi'\in\pctlal$ such that $\Phi\equiv\Phi'$.
  We proceed by structural induction on $\Phi$. 
  \begin{enumerate}
  \item $\Phi\equiv\Phi^a$. Trivial.
  \item $\Phi\equiv\Phi_1\land\Phi_2$.
    Let $\mathit{Suf}(P)$ be the set such that $\tree\in\mathit{Suf}(P)$ iff $\tree\in P$ and
    $\tree'\in P$ for all $\tree'\in\alltrees^\omega$ with $\tree$ being a suffix of $\tree'$.
    Then by Def.~\ref{def:absolute liveness}, $P$ is an absolute liveness property iff $P\equiv\mathit{Suf}(P)$.
    In case $\Phi$ is an absolute liveness property, we show that 
    $\mathit{Suf}(\Phi)=\mathit{Suf}(\Phi_1\land\Phi_2)=\mathit{Suf}(\Phi_1)\land\mathit{Suf}(\Phi_2)$.
    By the definition of $\mathit{Suf}$, $\mathit{Suf}(\Phi_1\land\Phi_2)\subseteq\mathit{Suf}(\Phi_1)\cap\mathit{Suf}(\Phi_2)$.
    We show the other direction. Let $\tree\in\mathit{Suf}(\Phi_1)\cap\mathit{Suf}(\Phi_2)$.
    Then $\tree\in\Phi_1$ and $\tree\in\Phi_2$, and for any $\tree'$ such that $\tree$ is a suffix of $\tree'$,
    we have $\tree'\in\Phi_1$ and $\tree'\in\Phi_2$, i.e., $\tree'\in\Phi_1\land\Phi_2$.
    By the definition of $\mathit{Suf}$, $\tree\in\mathit{Suf}(\Phi_1\land\Phi_2)$.
    Hence in case $\Phi$ is an absolute liveness property, $\Phi_1\land\Phi_2\equiv\mathit{Suf}(\Phi_1)\cap\mathit{Suf}(\Phi_2)$.
    Since $\mathit{Suf}(\Phi_1)$ and $\mathit{Suf}(\Phi_2)$ are absolutely live,
    there exists $\Phi'_1,\Phi'_2\in\pctlal$ such that $\Phi'_1\equiv\mathit{Suf}(\Phi_1)$ and
    $\Phi'_2\equiv\mathit{Suf}(\Phi_2)$ by induction hypothesis. 
    Thus $\Phi$ can be represented as an equivalent formula $\Phi'=\Phi'_1\land\Phi'_2\in\pctlal$.
    The case for $\Phi\equiv\Phi_1\lor\Phi_2$ can be proved in a similar way and is omitted here.
  \item $\Phi\equiv\P{\unrhd q}{\X\Phi_1}$ with $\unrhd\in\{>,\ge\}$. We distinguish:
    \begin{enumerate}
    \item $q>0$. Suppose $\Phi_1\not\equiv\top$, otherwise $\Phi\equiv\top\in\pctlal$. 
      Let $\tree\in\Phi$. Construct a PT $\tree'$
      such that $\tree$ is a suffix of $\tree'$ and the probability of going to some $\tree''\in\neg\Phi_1$ in
      the next step is arbitrarily large such that $\tree'\not\in\Phi$. Thus $\Phi$ is not an absolute liveness property.
    \item $\unrhd=>$, $q=0$, and $\Phi_1$ is not an absolute liveness property. 
      Let $\tree\in\Phi$, i.e., the probability of $\tree$ satisfying $\Phi_1$ in the next step
      is positive. Since $\Phi_1$ is not an absolute liveness property, there exists $\tree'$
      such that the probability satisfying $\Phi_1$ in the next step is 0 with $\tree$ being 
      a suffix of $\tree'$. Thus $\tree'\not\in\Phi$ and $\Phi$ is not an absolute liveness property.
    \end{enumerate}
  \item $\Phi\equiv\P{\unrhd q}{\Phi_1\U\Phi_2}$.  We distinguish:
    \begin{enumerate}
    \item $q>0$ and $\Phi_2$ is an absolute liveness property.
      By induction hypothesis, there exists $\Phi'_2\in\pctlal$ such that $\Phi_2\equiv\Phi'_2$.
      In this case $\Phi\equiv\Phi'_2\in\pctlal$. $\Phi'_2\subseteq\Phi$ is straightforward. We
      show the other direction. For any $\tree\in\Phi$, the probability of reaching
      nodes satisfying $\Phi'_2$ is positive, i.e., there exists $\tree'\in\Phi'_2$
      such that $\tree'$ is a suffix of $\tree$. Since $\Phi'_2$ is an absolute liveness property,
      $\tree\in\Phi'_2$.
    \item $q>0$, $\Phi_1$ is absolutely live, $\Phi_2$ is not absolutely live, and $\neg\Phi_1\land\Phi_2\equiv\bot$.
      Assume $\P{=0}{\Phi_1\U\Phi_2}\not\equiv\bot$, otherwise $\Phi\equiv\top\in\pctlal$.
      Since $\Phi_2$ is not an absolute liveness property,
      for $\tree\in\Phi_2$, there exists $\tree'\in\alltrees^\omega$ such that $\tree$ is a suffix of $\tree'$
      and in $\tree'$ the probability of satisfying
      $\Phi_1\U\Phi_2$ is arbitrarily small (by making the probability of the transitions to some $\tree''\in\P{=0}{\Phi_1\U\Phi_2}$ great enough)
      such that $\tree'\not\in\Phi$. Since $\tree\in\Phi_2$ implies
      $\tree\in\Phi$, $\Phi$ is not an absolute liveness property.
    \item Assume neither $\Phi_1$ nor $\Phi_2$ is absolutely live.
      Let $\tree\in\Phi_2$ which implies $\tree\in\Phi$.
      Since $\Phi_2$ is not an absolute liveness property. There exists $\tree'\in\neg\Phi_2$
      such that $\tree$ is a suffix of $\tree'$.
      In case $\tree'\in\neg\Phi_1$, then $\tree'\not\in\Phi$, which indicates that $\Phi$
      is not an absolute liveness property.
      In case such $\tree'$ does not exist, which indicates that once $\tree\in\Phi_2$, then
      $\tree''\in\Phi_1\lor\Phi_2$ for all $\tree''\in\alltrees^\omega$ such that
      $\tree$ is a suffix of $\tree''$. 
      If $\unrhd=>$ and $q=0$, then $\Phi\equiv\Phi'=\P{>0}{\Diamond\Phi_2}\in\pctlal$
      by Def.~\ref{def:absolute liveness}.
      The proof when $q>0$ is similar as the above case.
    \item $\Phi_1$ is absolutely live, $\Phi_2$ is not absolutely live, and $\neg\Phi_1\land\Phi_2\not\equiv\bot$.
      Let $\tree\in\neg\Phi_1\land\Phi_2$, which implies $\tree\in\Phi$.
      By induction hypothesis, $\Phi_1$ is an absolute liveness, while $\Phi_2$ is not.
      There exists $\tree'\in\neg\Phi_2$ such that $\tree$ is a suffix of $\tree'$.
      By Lemma~\ref{lem:stable and absolute}, $\neg\Phi_1$ is a stable property.
      Thus $\tree'\in\neg\Phi_1$ by Def.~\ref{def:stable}. Since $\tree'\in\neg\Phi_1\land\neg\Phi_2$
      implies $\tree'\not\in\Phi$. We conclude that $\Phi$ is not an absolute liveness property.
    \end{enumerate}
  \item $\Phi\equiv\P{\unrhd q}{\Phi_1\W\Phi_2}$. We distinguish:
    \begin{enumerate}
    \item $q>0$ and $\Phi_1,\Phi_2$ are absolutely live. Let $\tree\in\Phi$ such that
      $\tree\in\P{\unrhd q}{\Box(\Phi_1\land\neg\Phi_2)}$. If such $\tree$
      does not exist, i.e., for any $\tree\in\Phi$, $\tree\in\P{>0}{\Diamond\Phi_2}$.
      Since $\Phi_2$ is absolutely live, $\Phi\equiv\Phi_2$ (any PT in $\P{>0}{\Diamond\Phi_2}$ must have a
      suffix in $\Phi_2$).
      Moreover, there exists $\Phi'_2\in\pctlal$ such that $\Phi_2\equiv\Phi'_2$ by induction hypothesis.
      Hence $\Phi\equiv\Phi'_2\in\pctlal$.
      Note $\neg\Phi_1\land\neg\Phi_2\not\equiv\bot$, otherwise $\Phi\equiv\top\in\pctlal$.
      For a PT $\tree\in\P{\unrhd q}{\Box(\Phi_1\land\neg\Phi_2)}$, there always exists $\tree'$ such
      that $\tree$ is a suffix of $\tree'$ and the probability of $\tree'$ satisfying $\Box(\Phi_1\land\neg\Phi_2)$
      is arbitrarily small (by making the probability from $\tree'$ to $\tree$ arbitrarily small, while all other transitions lead to a
      PT in $\neg\Phi_1\land\neg\Phi_2$) such that $\tree'\not\in\Phi$. Thus $\Phi$ is not an absolute liveness property.
    \item $q>0$, $\Phi_1$ is absolutely live, $\Phi_2$ is not absolutely live, and $\neg\Phi_1\land\Phi_2\equiv\bot$.
      It can be proved in a similar way as for $\U$. Let $\tree\in\Phi$, we
      can construct $\tree'$ with $\tree$ being a suffix of $\tree'$ and the probability of $\tree'$ satisfying
      $\Phi_1\W\Phi_2$ is arbitrarily small.
    \item $\Phi_1$ is not absolutely live.
      We note that $\Phi\equiv\P{\unrhd q}{\Phi_1\U(\P{\ge 1}{\Box\Phi_1}\lor\Phi_2)}$.
      Since $\Phi_1$ is not absolutely live, $\Phi_1\not\equiv\top$.
      Directly from Def.~\ref{def:absolute liveness}, $\P{\ge 1}{\Box\Phi_1}$ is not absolutely live either. 
      According to the above proof for $\U$ modality,
      for $\Phi$ being absolutely live, it must be the case
      that $\Phi_2$ is an absolute liveness property and $\P{\ge 1}{\Box\Phi_1}\lor\Phi_2\equiv\Phi_2$,
      which implies $\Phi\equiv\P{\unrhd q}{\Phi_1\U\Phi_2}\equiv\P{>0}{\Phi_1\U\Phi_2}\in\pctlal$.
    \item $\Phi_1$ is absolutely live, $\Phi_2$ is not absolutely live, and $\neg\Phi_1\land\Phi_2\not\equiv\bot$.
      By Def.~\ref{def:absolute liveness}, we can show
      that $\P{\ge 1}{\Box\Phi_1}$ is not an absolute liveness property, hence neither is
      $\P{\ge 1}{\Box\Phi_1}\lor\Phi_2$.
      Again by making use of the fact that $\Phi\equiv\P{\unrhd q}{\Phi_1\U(\P{\ge 1}{\Box\Phi_1}\lor\Phi_2)}$,
      for $\Phi$ to be an absolute liveness property, it must be the case that $\neg\Phi_1\land(\P{\ge 1}{\Box\Phi_1}\lor\Phi_2)\equiv\bot$.
      However, $\neg\Phi_1\land(\P{\ge 1}{\Box\Phi_1}\lor\Phi_2)\equiv(\neg\Phi_1\land\P{\ge 1}{\Box\Phi_1})\lor(\neg\Phi_1\land\Phi_2)\not\equiv\bot$.
      Thus $\Phi$ is not an absolute liveness property.
    \end{enumerate}
  \item All other cases are either simple or can be proved using duality laws.
  \end{enumerate}
\vspace*{-0.5cm}
\end{proof}

\begin{reptheorem}{thm:absolute characterisation}
\emph{PCTL}-formula $\Phi$ is an absolute liveness property iff
$\Phi\not\equiv\bot$ and $\Phi\equiv\P{>0}{\Diamond\Phi}.$
\end{reptheorem}
\begin{proof} 
  Let $\Phi\not\equiv\bot$ be a PCTL formula.
  \begin{enumerate}
  \item $\Phi\equiv\P{>0}{\Diamond\Phi}$. We prove that $\Phi$ is an absolute liveness property.
    By contraposition. Suppose $\Phi$ is not an absolute liveness property. Then
    there exists $\tree\in\Phi$ and $\tree'\not\in\Phi$ such that 
    $\tree$ is a suffix of $\tree'$. Due to that $\tree$ is a suffix of $\tree'$,
    the probability of reaching $\tree$ is positive, i.e., $\tree\in\P{>0}{\Diamond\Phi}$.
    Contradiction.
  \item $\Phi$ is an absolute liveness property. We prove that $\Phi\equiv\P{>0}{\Diamond\Phi}$.
    Obviously, $\Phi\subseteq\P{>0}{\Diamond\Phi}$, thus we only show that $\P{>0}{\Diamond\Phi}\subseteq\Phi$.
    By contraposition. Suppose there is $\tree\in\alltrees^\omega$ such that 
    $\tree\in\P{>0}{\Diamond\Phi}$ and $\tree\not\in\Phi$.
    Since $\tree\in\P{>0}{\Diamond\Phi}$, the probability of $\tree$ reaching its suffixes in $\Phi$
    is positive. In other words, there exists $\tree'\in\Phi$, where $\tree'$ is a suffix
    of $\tree$. This contradicts that $\Phi$ is an absolute liveness property. 
  \end{enumerate}
\vspace*{-0.5cm}
\end{proof}


\end{document}